\documentclass[a4paper,onecolumn,accepted=2021-02-09]{quantumarticle}
\pdfoutput=1

\usepackage{graphicx}
\usepackage{xspace}
\usepackage{amsmath}
\usepackage{amsthm}
\usepackage{amsfonts}
\usepackage{amssymb}
\usepackage{afterpage}
\usepackage{wrapfig}
\usepackage{framed}
\usepackage{esvect}
\usepackage{graphics}
\usepackage{color}
\usepackage{xcolor}
\usepackage{verbatim}
\usepackage{framed}
\usepackage{microtype}
\usepackage{hyperref}
\usepackage{float}
\usepackage{mathrsfs} 
\usepackage[normalem]{ulem}
\usepackage{etoolbox}
\usepackage{fullpage}

\usepackage[backend=bibtex,style=alphabetic,maxbibnames=99]{biblatex}
\newcommand{\ignore}[1]{}

\definecolor{darkred}{rgb}{0.5, 0, 0}
\definecolor{darkgreen}{rgb}{0, 0.5, 0}
\definecolor{darkblue}{rgb}{0,0,0.5}

\newlength{\saveparindent}
\newlength{\saveparskip}
\newcounter{ctr}

\newcounter{ectr}

\newenvironment{tiret}{%
\begin{list}{\hspace{1pt}\rule[0.5ex]{6pt}{1pt}\hfill}{\labelwidth=15pt%
\labelsep=3pt \leftmargin=18pt \topsep=1pt%
\setlength{\listparindent}{\saveparindent}%
\setlength{\parsep}{\saveparskip}%
\setlength{\itemsep}{1pt}}}{\end{list}}

\newcommand{\X}{\spa{X}}
\newcommand{\Y}{\spa{Y}}

\newcommand{\abs}[1]{\ensuremath{\lvert{#1}\rvert}}

\newcommand{\set}[1]{\ensuremath{\{#1\}}}

\newcommand{\Party}[1][\relax]{\ensuremath{P_{#1}}\xspace}

\newcommand{\PartyA}{\mathscr{A}}
\newcommand{\PartyB}{\mathscr{B}}

\newcommand{\PS}{\Party[\sf s]}
\newcommand{\PR}{\Party[\sf r]}

\newcommand{\Sim}{\ensuremath{\mathcal{S}}\xspace}

\newcommand{\Oracle}{\ensuremath{\mathcal{O}}\xspace}

\newcommand{\secp}{\ensuremath{n}\xspace}

\newcommand{\cF}{\ensuremath{\mathcal{F}}\xspace}
\newcommand{\cA}{\ensuremath{\mathcal{A}}\xspace}
\newcommand{\cZ}{\ensuremath{\mathcal{Z}}\xspace}
\newcommand{\cS}{\ensuremath{\mathcal{S}}\xspace}

\newcommand{\create}{\ensuremath{\mathsf{create}}\xspace}
\newcommand{\run}{\ensuremath{\mathsf{run}}\xspace}

\newcommand{\poly}{\ensuremath{\mathrm{poly}}\xspace}

\newcommand{\lmax}{\lambda_{\operatorname{max}}}

\newcommand{\rejz}{A_{\overline{0}}}
\newcommand{\rejo}{A_{\overline{1}}}
\newcommand{\accz}{A_{0}}
\newcommand{\acco}{A_{1}}

\usepackage{stackengine}
\newcommand\oast{\stackMath\mathbin{\stackinset{c}{0ex}{c}{0ex}{\ast}{\bigcirc}}}

\newcommand{\func}[1][\relax]{{\ensuremath{\mathcal{F}_{\tt #1}}}\xspace}
\newcommand{\funcOTM}{\ensuremath{\func[OTM]}\xspace}

\newcommand{\funcHT}{\ensuremath{\func[wrap]}\xspace}

\newcommand{\cH}{\ensuremath{\mathcal{H}}\xspace}

\newcommand{\pidum}{P}

\def\exec{\mathsf{EXEC}}

\newtheorem{theorem}{Theorem}[section]
\newtheorem{definition}[theorem]{Definition}

\newtheorem{remark}[theorem]{Remark}

\newtheorem{lemma}[theorem]{Lemma}

\newtheorem{conj}[theorem]{Conjecture}

\newtheoremstyle{named}{}{}{\itshape}{}{\bfseries}{.}{.5em}{\thmnote{#3's }#1}
\theoremstyle{named}
\newtheorem*{namedtheorem}{Main Theorem (informal)}

\newlength{\protowidth}

\newcommand{\namedref}[2]{\hyperref[#2]{#1~\ref*{#2}}}

\newcommand{\subfigureref}[2]{\hyperref[#1]{Figure~\ref*{#1}#2}}

\definecolor{darkred}{rgb}{0.5, 0, 0}
\definecolor{darkgreen}{rgb}{0, 0.5, 0}
\definecolor{darkblue}{rgb}{0,0,0.5}

\newcommand{\bra}[1]{\langle #1|}
\newcommand{\ket}[1]{|#1\rangle}
\newcommand{\braket}[2]{\langle #1|#2\rangle}
\newcommand{\ketbra}[2]{\ket{#1}{\bra{#2}}}

\newcommand{\Tr}{Tr} 

\newcommand{\unitary}{\mathcal{U}}
\newcommand{\trace}{{\rm Tr}}

\renewcommand{\set}[1]{{\left\{#1\right\}}}    

\renewcommand{\abs}[1]{\left\lvert #1 \right\rvert}

\newcommand{\hgate}{\operatorname{H}}
\renewcommand{\poly}{\operatorname{poly}}

\newcommand{\complex}{{\mathbb C}}

\renewcommand{\comment}[1]{}

\newcommand{\spa}[1]{\mathcal{#1}}
\newcommand{\dens}{\mathcal{D}}

\mathchardef\mhyphen="2D

\newcommand{\psigood}{\ket{\psi_{\rm good}}}
\newcommand{\pigood}{\Pi_{\rm good}}
\newcommand{\psiapprox}{\ket{\psi_{\rm approx}}}

\usepackage{color}

\floatstyle{ruled}
\newfloat{protocol}{tph}{lop}
\floatname{protocol}{Protocol}
\newfloat{program}{tph}{lop}
\floatname{program}{Program}

\newfloat{scheme}{tph}{lop}
\floatname{scheme}{Scheme}
\newfloat{functionality}{htp}{lop}
\floatname{functionality}{Functionality}
\newfloat{simulator}{tph}{lop}
\floatname{simulator}{Simulator}
\newfloat{experiment}{tph}{lop}
\floatname{experiment}{Experiment}

\begin{document}

\title{Towards Quantum One-Time Memories from Stateless Hardware}

\author{Anne Broadbent}
\affiliation{Department of Mathematics and Statistics, University of Ottawa, Ontario, Canada}
\email{abroadbe@uottawa.ca}
\author{Sevag Gharibian}
\affiliation{Department of Computer Science, Paderborn University, Germany, and Virginia Commonwealth University, USA}
\email{sevag.gharibian@upb.de}
\author{Hong-Sheng Zhou}
\affiliation{Department of Computer Science, Virginia Commonwealth University, Virginia, USA}
\email{hszhou@vcu.edu}

\maketitle


\begin{abstract}
A central tenet of theoretical cryptography is the study of the minimal assumptions required to implement a given cryptographic primitive. One such primitive is the one-time memory~(OTM), introduced by Goldwasser, Kalai, and Rothblum [CRYPTO 2008], which is a classical functionality modeled after a non-interactive \mbox{1-out-of-2} oblivious transfer, and which is complete for one-time classical and quantum programs. It is known that secure OTMs do not exist in the standard model in both the classical and quantum settings. Here, we propose a scheme for using quantum information, together with the assumption of stateless (\emph{i.e.},  reusable) hardware tokens, to build statistically secure OTMs. Via the semidefinite programming-based quantum games framework of Gutoski and Watrous [STOC 2007], we prove security for a malicious receiver making at most $0.114n$ adaptive queries to the token (for $n$ the key size), in the quantum universal composability framework, but leave open the question of security against a polynomial amount of queries. Compared to alternative schemes derived from the literature on quantum money, our scheme is technologically simple since it is of the ``prepare-and-measure'' type. We also give two impossibility results showing certain assumptions in our scheme cannot be relaxed.
\end{abstract}

\section{Introduction}
Theoretical cryptography centers around building cryptographic primitives secure against adversarial attacks.
In order to allow a broader set of such primitives to be implemented, one often considers restricting the power of the adversary.  For example, one can limit the {\em computing} power of adversaries to be polynomial bounded~\cite{FOCS:Yao82a,FOCS:BluMic82}, restrict the {\em storage} of adversaries to be bounded or noisy~\cite{Mau92,CM97,DFSS05}, or make {\em trusted setups} available to honest players \cite{STOC:Kilian88,STOC:BluFelMic88,FOCS:Canetti01,STOC:CLOS02,C:IshPraSah08,C:PraRos08,STOC:LinPasVen09,TCC:MajPraRos09,C:MajPraRos10,ITCS:MauRen11,TCC:KraMul11,EC:KMPS14}, to name a few. One well-known trusted setup  is \emph{tamper-proof hardware}~\cite{EC:Katz07,C:GolKalRot08}, which is assumed to provide a specific input-output functionality, and which can only be accessed in a ``black box'' fashion. 
The hardware can maintain a state (\emph{i.e.}, is \emph{stateful}) and possibly carry out complex functionality, but presumably may be difficult or expensive to implement or manufacture. This leads to an interesting research direction: Building cryptography primitives using the {\em simplest} (and hence easiest and cheapest to manufacture) hardware.

In this respect, two distinct simplified notions of hardware have captured considerable interest.
The first is the notion of a {\em one-time memory (OTM)}~\cite{C:GolKalRot08}, which
is arguably the simplest possible notion of {\em stateful} hardware. An OTM, modeled after a non-interactive 1-out-of-2 {oblivious transfer}, behaves as follows:
first, a player (called the \emph{sender}) embeds two values $s_0$ and~$s_1$ into the OTM, and then gives the OTM to another player (called the \emph{receiver}). The receiver can now read his choice of precisely one of $s_0$ or~$s_1$; after this ``use'' of the OTM, however, the unread bit is lost forever.
Interestingly, OTMs are complete for implementing \emph{one-time} use programs (OTPs): 
 given access to OTMs, one can implement statistically secure OTPs for any efficiently computable program in the universal composability (UC) framework~\cite{TCC:GISVW10}. (OTPs, in turn, have applications in software protection and one-time proofs~\cite{C:GolKalRot08}.) In the {quantum} UC model, OTMs enable \emph{quantum} one-time programs~\cite{C:BroGutSte13}.  (This situation is analogous to the case of \emph{oblivious transfer} being complete for two-party secure function evaluation~\cite{STOC:Kilian88,C:IshPraSah08}.)
Unfortunately, OTMs are inherently \emph{stateful}, and thus represent a very strong cryptographic assumption --- any physical implementation of such a device must somehow maintain internal knowledge between activations, \emph{i.e.}, it must completely ``self-destruct'' after a single use.

This brings us to a second important simplified notion of hardware known as a \emph{stateless} token~\cite{EC:ChaGoySah08}, which keeps
no record of previous interactions. 
On the positive side, such hardware is presumably easier to implement. On the negative side, an adversary can run an experiment with stateless hardware as many times as desired, and each time the hardware is essentially ``reset''. (Despite this, stateless hardware has been useful in achieving {\em computationally secure} multi-party computation~\cite{EC:ChaGoySah08,TCC:GISVW10,TCC:CKSYZ14}, and  {\em statistically secure} commitments~\cite{AC:DamSca13}.) It thus seems impossible for stateless tokens to be helpful in implementing any sort of ``self-destruct'' mechanism.  Indeed, classically stateful tokens are trivially more powerful than stateless ones, as observed in, \emph{e.g.},~\cite{TCC:GISVW10}. This raises the question:
\begin{quote}
{\em
Can \emph{quantum} information, together with a classical stateless token, be used to simulate ``self destruction'' of a hardware token?
}
\end{quote}

\noindent In particular, a natural question along these lines is whether quantum information can help implement an~OTM. Unfortunately, it is known that quantum information \emph{alone} cannot implement an~OTM (or, more generally, any one-time program)~\cite{C:BroGutSte13}; see also Section~\ref{sec:Impossibility} below. We thus ask the question: What are the minimal cryptographic assumptions required in a quantum world to implement an~OTM?

\paragraph{Contributions and summary of techniques.} We propose what is, to our knowledge, the first prepare-and-measure quantum protocol that constructs  OTMs from stateless hardware tokens. For this protocol, we are able to rigorously prove information theoretic security against an adversary making a \emph{linear} (in~$n$, the security parameter) number of adaptive queries to the token. While we conjecture that security holds also for \emph{polynomially} many queries, note that already in this setting of linearly many adaptive queries, our protocol achieves something impossible classically (\emph{i.e.}, classically, obtaining security against a linear number of queries is impossible). We also show stand-alone security against a malicious sender.

\smallskip
\noindent

\noindent\underline{\sc Historical Note.} We proposed the concept that quantum information could provide a ``stateless to stateful'' transformation in a preliminary version of this work~\cite{BGZ15}; however, that work claimed security against a \emph{polynomial} number of token queries, obtained via a reduction from the interactive to the non-interactive setting. We thank an anonymous referee for catching a subtle, but important bug which appears to rule out the proof approach of~\cite{BGZ15}. The current paper hence employs a different proof approach, which models interaction with the token as a ``quantum game'' via semidefinite programming (further details below). Since our original paper was posted, recent work~\cite{CGLZ18} has shown an alternate quantum ``stateful to stateless'' transformation via quantum money constructions~\cite{BS18}. Specifically, in~\cite{CGLZ18}, security against a polynomial number of queries is achieved, albeit with respect to a new definition  of ``OTMs relative to an oracle'' (while the security results of the present paper are with respect to the well-established simulation-based definition of \cite{TCC:GISVW10,EC:Katz07}).  Furthermore,~\cite{CGLZ18} directly applies known quantum money constructions, which require difficult-to-prepare highly entangled states. Our focus here, in contrast, is to take a ``first-principles'' approach and build a technologically simple-to-implement scheme which requires no entanglement, but rather the preparation of just one of four single qubit states, $\ket{0},\ket{1},\ket{+},\ket{-}$. Indeed, the two works are arguably complementary in that the former focuses primarily on \emph{applications} of ``stateful'' single-use tokens, while our focus is on the most technologically simple way to \emph{implement} such ``stateful'' tokens.

\smallskip
\noindent
\underline{\sc Construction.} Our construction is inspired by Wiesner's \emph{conjugate coding}~\cite{wiesner1983conjugate}: the quantum portion of the protocols consists in $n$  quantum states chosen uniformly at random from $\{\ket{0}, \ket{1}, \ket{+}, \ket{-}\}$ (note this encoding is independent of the classical bits of the OTM functionality). We then couple this $n$-qubit quantum state, $\ket{\psi}$ (the \emph{quantum key})
with a \emph{classical} stateless hardware token, which takes as inputs a choice bit $b$, together with an $n$-bit string~$y$. If $b=0$, the hardware token verifies that the bits of $y$ that correspond to \emph{rectilinear} ($\ket{0}$ or $\ket{1}$, \emph{i.e.}, $Z$ basis) encoded qubits of $\ket{\psi}$ are consistent with the measurement of $\ket{\psi}$ in the computational basis, in which case the bit $s_0$ is returned. If $b=1$, the hardware token verifies that the bits of $y$ that correspond to \emph{diagonal} ($\ket{+}$ or $\ket{-}$, \emph{i.e.}, $X$ basis) encoded qubits of $\ket{\psi}$ are consistent with the measurement of $\ket{\psi}$ in the diagonal  basis, in which case the bit $s_1$ is returned.\footnote{We note that a simple modification using a classical one-time pad could be used to make \emph{both} the quantum state and hardware token independent of $s_0$ and $s_1$: the token would output one of two uniformly random bits $r_0$ and $r_1$, which could each be used to decrypt a single bit, $s_0$ or $s_1$.}
The honest use of the OTM is thus intuitive: for choice bit $b=0$, the user measures each qubit of the quantum key in the rectilinear basis to obtain an $n$-bit string~$y$, and inputs $(b,y)$ into the hardware token. If $b=1$, the same process is applied, but with measurements in the diagonal basis.

\smallskip
\noindent
\underline{\sc Assumption.} Crucially, we assume the hardware token accepts \emph{classical} input only (alternatively and equivalently, the token immediately measures its quantum input in the standard basis), \emph{i.e.}, it cannot be queried in superposition. Although this may seem a strong assumption, in Section~\ref{sscn:super} we show that any token which can be queried in superposition {in a reversible way}, cannot be used to construct a secure OTM (with respect to our setting in which the adversary is allowed to apply arbitrary quantum operations). Similar classical-input hardware has previously been considered in, \emph{e.g.},~\cite{C:Unruh13,C:BroGutSte13}.

\smallskip

\noindent
\underline{\sc Security and intuition.}
Stand-alone security against a malicious sender is relatively straightforward to establish, since the protocol consists in a single message from the sender to the receiver, and since stand-alone security only requires simulation of the \emph{local} view of the adversary.

The intuition underlying security against a malicious receiver is clear: in order for a receiver to extract a bit $s_b$ as encoded in the OTM, she must perform a complete measurement of the qubits of $\ket{\psi}$ in order to obtain a classical {key} for $s_b$ (since, otherwise, she would likely fail the test as imposed by the hardware token). But such a  measurement would invalidate the receiver's chance of extracting the bit $s_{1- b}$! This is exactly the ``self-destruct''-like property we require in order to implement an OTM. This intuitive notion of security was present in Wiesner's proposal for quantum money~\cite{wiesner1983conjugate}, and is often given a physical explanation in terms of the no-cloning theorem~\cite{wootters1982single} or Heisenberg uncertainty relation~\cite{Hei27}.

Formally, we work in the  statistical (\emph{i.e.}, information-theoretic) setting of the quantum \emph{Universal Composability} (UC) framework~\cite{EC:Unruh10}, which allows us to make strong security statements that address the \emph{composability} of our protocol within others. As a proof technique, we describe a simulator, such that for any ``quantum environment'' wishing to interact with the OTM,
 the environment statistically cannot tell whether it is interacting with the \emph{ideal} OTM functionality or the \emph{real} OTM instance provided by our scheme. The security of this simulator requires a statement of the following form: Given access to a (randomly chosen) ``quantum key'' $\ket{\psi_k}$ and corresponding stateless token~$V_k$, it is highly unlikely for an adversary to successfully extract keys for \emph{both} the secret bits $s_0$ and $s_1$ held by $V_k$. We are able to show this statement for any adversary which makes a linear number of queries, by which we mean an adversary making $m$ queries succeeds with probability at most $O(2^{2m-0.228n})$ (for $n$ the number of quantum key bits in $\ket{\psi_k}$). In other words, if the adversary makes at most $m=cn$ queries with $c<0.114$, then its probability of cheating successfully is exponentially small in $n$. We conjecture, however, that a similar statement holds for any $m\in\poly(n)$, \emph{i.e.}, that the protocol is secure against polynomially many queries.

 To show security against linearly many queries, we exploit the semidefinite programming-based quantum games framework of Gutoski and Watrous (GW)~\cite{GutoskiW07} to model interaction with the token. Intuitively, GW is useful for our setting, since it is general enough to model multiple rounds of {adaptive} queries to the token, even when the receiver holds quantum ``side information'' in the form of $\ket{\psi}$. We describe this technique in Sections~\ref{sscn:gw} and~\ref{sscn:securityintuition}, and provide full details in Appendix~\ref{scn:newproof}. Summarizing, we show the following.

\begin{namedtheorem}
There exists a protocol $\Pi$, which together with a classical stateless token and the ability to randomly prepare single qubits in one of four pure states, implements the OTM functionality with statistical security in the UC framework against a corrupted receiver making at most $cn$ queries for any $c<0.114$.
\end{namedtheorem}

\noindent As stated above, we conjecture that our protocol is actually secure against polynomially many adaptive queries. However, we are unable to show this claim using our present proof techniques, and hence leave this question open. Related to this, we make the following comments: (1) As far as we are aware, the Main Theorem above is the only known formal proof of any type of security for conjugate coding in the interactive setting with $\Omega(1)$ queries. Moreover, as stated earlier, classical security against $\Omega(1)$ queries is trivially impossible. (2) Our proof introduces the GW semidefinite programming framework from quantum interactive proofs to the study of conjugate coding-based schemes. This framework allows  handling multiple challenges in a unified fashion: arbitrary quantum operations by the user, classical queries to the token, and the highly non-trivial assumption of quantum side information for the user (the ``quantum key'' state sent to the user.)\\

\noindent\emph{Towards security against polynomially many queries.} Regarding the prospects of proving security against polynomially many adaptive queries, we generally believe it requires a significant new insight into how to design a ``good'' feasible solution to the primal semidefinite program (SDP) obtained via GW. However, in addition to our proof for linear security (Theorem~\ref{thm:cheatingbound}), in Appendix~\ref{app:cleaner} we attempt to give evidence towards our conjecture for polynomial security. Namely, Appendix~\ref{sscn:stream} first simplifies the SDPs obtained from GW, and derives the corresponding dual SDPs. We remark these derivations apply for any instantiation of the GW framework, \emph{i.e.} they are not specific to our setting, and hence may prove useful elsewhere. In Appendix~\ref{sscn:approx}, we then give a feasible solution $Y$ (Equation~(\ref{eqn:prob})) to the dual SDP. While $Y$ is simple to state, it is somewhat involved to analyze. A heuristic analysis suggests $Y$'s dual objective function value has roughly the behavior needed to show security, \emph{i.e.} the value scales as $m/\sqrt{2^n}$, for $m$ queries and $n$ key bits. If $Y$ were to be the \emph{optimal} solution to the dual SDP, this would strongly suggest the optimal cheating probability is essentially $m/\sqrt{2^n}$. However, we explicitly show $Y$ is not optimal, and so $m/\sqrt{2^n}$ is only a \emph{lower bound} on the optimal cheating probability\footnote{Indeed, an attack in the Breidbart basis breaks our scheme with probability $2^{-0.228n}$, as observed by David Mestel; see Section~\ref{sscn:approx}.}. Nevertheless, we conjecture that while $Y$ is not optimal, it is \emph{approximately} optimal (see Conjecture~\ref{conj:only} for a precise statement); this would imply the desired polynomial security claim. Unfortunately, the only techniques we are aware of to show such approximate optimality involve deriving a better primal SDP solution, which appears challenging.\\

\noindent{\bf Further Related work.}
Our work contributes to the growing list of functionalities achievable with quantum information, yet unachievable classically. This includes: unconditionally secure key expansion \cite{BB84}, physically uncloneable money~\cite{wiesner1983conjugate,MVW13,PYJLC12},  a  reduction from oblivious transfer to bit commitment~\cite{C:BBCS91,C:DFLSS09} and to other primitives such as ``cut-and choose'' functionality~\cite{TCC:FKSZZ13}, and   revocable time-release quantum encryption~\cite{EC:Unruh14}. Importantly, these protocols all make use of the technique of conjugate coding~\cite{wiesner1983conjugate}, which is also an important technique used in protocols for OT in the bounded quantum storage and noisy quantum storage models~\cite{DFSS05,WST08} (see~\cite{BS15} for a survey).

A number of proof techniques have been developed in the context of conjugate coding, including entropic uncertainty relations~\cite{WW10}. In the context of QKD, another technique is the use of de Finetti reductions~\cite{Ren08} (which exploit the symmetry of the scheme in order to simplify the analysis).  Recently, semidefinite programming (SDP) approaches have been applied to analyze security of conjugate coding~\cite{MVW13} for quantum money, in the setting of one round of interaction with a ``stateful'' bank. SDPs are also the technical tool we adopt for our proof (Section~\ref{sscn:securityintuition} and Appendix~\ref{scn:newproof}), though here we require the more advanced quantum games SDP framework of Gutoski and Watrous~\cite{GutoskiW07} to deal with {multiple} adaptive interactions with {stateless} tokens. Reference~\cite{PYJLC12} has also made use of Gavinsky's~\cite{G12} quantum retrieval games framework.

Continuing with proof techniques, somewhat similar to~\cite{PYJLC12}, Aaronson and Christiano~\cite{AC12} have studied quantum money schemes in which one interacts with a verifier. They introduce an ``inner product adversary method'' to lower bound the number of queries required to break their scheme.

We remark that \cite{PYJLC12} and~\cite{MVW13} have studied schemes based on conjugate coding similar to ours, but in the context of quantum money. In contrast to our setting, the schemes of~\cite{PYJLC12} and~\cite{MVW13} (for example) involve dynamically chosen random challenges from a verifier to the holder of a ``quantum banknote'', whereas in our work here the ``challenges'' are fixed (\emph{i.e.}, measure all qubits in the $Z$ or $X$ basis to obtain secret bit $s_0$ or $s_1$, respectively), and the verifier is replaced by a stateless token. Thus, ~\cite{MVW13}, for example, may be viewed as using a ``stateful'' verifier, whereas our focus here is on a ``stateless'' verifier (\emph{i.e.},~a token).

Also, we note that prior work has achieved oblivious transfer using quantum information, together with some assumption (\emph{e.g.},~bit commitment~\cite{C:BBCS91} or bounded quantum storage~\cite{DFSS05}). These protocols typically use an interaction phase similar to the ``commit-and-open'' protocol of~\cite{C:BBCS91}; because we are working in the non-interactive setting, these techniques appear to be inapplicable.

Finally, Liu~\cite{ITCS:Liu14,C:Liu14,EC:Liu15} has given stand-alone secure OTMs using quantum information in the \emph{isolated-qubit model}. Liu's approach is nice in that it avoids the use of trusted setups. In return, however, Liu must use the {isolated-qubit model}, which restricts the adversary to perform only single-qubit operations (no entangling gates are permitted); this restriction is, in some sense,  necessary if one wants to avoid trusted setups, as a secure OTM in the plain quantum model cannot exist (see Section~\ref{sec:Impossibility}). In contrast, in the current work we allow unbounded and unrestricted quantum adversaries, but as a result require a trusted setup.  In addition, we remark the security notion of OTMs of~\cite{ITCS:Liu14,C:Liu14,EC:Liu15} is weaker than the simulation-based notion studied in this paper, and it remains an interesting open question whether the type of OTM in~\cite{ITCS:Liu14,C:Liu14,EC:Liu15} is secure under composition (in the current work, the UC framework gives us security under composition for free).

\medskip

\noindent{\bf Significance.} Our results show a strong separation between the classical and quantum settings, since classically, stateless tokens cannot be used to securely implement OTMs.
  To the best of our knowledge, our work is the first to combine conjugate coding with \emph{stateless} hardware tokens. Moreover, while our protocol shares similarities with prior work in the setting of quantum money, building OTMs appears to be a new focus here \footnote{We remark, however, that a reminiscent concept of single usage of quantum ``tickets'' in the context of quantum money is very briefly mentioned in Appendix S.4.1 of~\cite{PYJLC12}.}.

Our protocol has a simple implementation, fitting into the single-qubit prepare-and-measure paradigm, which is widely used as the ``benchmark'' for a ``physically feasible'' quantum protocol (in this model, one needs only the ability to prepares single-qubit states $\ket{0},\ket{1},\ket{+},\ket{-}$, and to perform single-qubit projective measurements. In particular, no entangled states are required, and in principle no quantum memory is required, since qubits can be measured one-by-one as they arrive).
In addition, from a theoretical cryptographic perspective, our protocol is attractive in that its implementation
requires an assumption of a stateless hardware token, which is conceivably easier and cheaper to manufacture (e.g. analogous to an RFID tag) than a stateful token.

In terms of security guarantees, we allow \emph{arbitrary} operations on behalf of a malicious quantum receiver in our protocol (\emph{i.e.}, all operations allowed by quantum mechanics), with the adversary restricted in that the stateless token is assumed only usable as a black box. The security we obtain is statistical, with the only computational assumption being on the number of \emph{queries} made to the token (recall we show security for a linear number of queries, and conjecture security for polynomially many queries). Finally, our security analysis is in the quantum UC framework against a corrupted receiver; this means our protocol can be easily composed with many others; for example, combining our results with~\cite{C:BroGutSte13}'s protocol immediately yields UC-secure quantum OTPs against a dishonest receiver.

We close by remarking that our scheme is ``tight'' with respect to two impossibility results, both of which assume the adversary has black-box access to both the token and its inverse operation\footnote{This is common in the oracle model of quantum computation, where a function $f:\set{0,1}^n\mapsto\set{0,1}$ is implemented via the (self-inverse) unitary mapping  $U_f\ket{x}\ket{y}=\ket{x}\ket{y\oplus f(x)}$.}. First, the assumption that the token be queried only in the computational basis cannot be relaxed: Section~\ref{sscn:super} shows that if the token can be queried in superposition, then an adversary in our setting can easily break any OTM scheme. Second, our scheme has the property that corresponding to each secret bit $s_i$ held by the token, there are exponentially many valid keys  one can input to the token to extract $s_i$. In Section~\ref{sscn:bounded}, we show that for any ``measure-and-access'' OTM  (\emph{i.e.}, an OTM in which one measures a given quantum key and uses the classical measurement result to access a token to extract data,  of which our protocol is an example\footnote{The term ``measure-and-access'' here is not to be confused with ``prepare-and-measure''. We define the former in Section~\ref{sscn:bounded} to mean a protocol in which one measures a given quantum resource state to extract a classical key, which is then used for a desired purpose. ``Prepare and measure'', in contrast, is referring to the fact that our scheme is easy to implement; the preparer of the token just needs to prepare single-qubit states, and an honest user simply measures them.}), a polynomial number of keys implies the ability to break the scheme with inverse polynomial probability (more generally, $\Delta$ keys allows probability at least $1/\Delta^2$ of breaking the scheme).

\medskip

\noindent{\bf Open Questions.}
While our work shows the fundamental advantage that quantum information yields in a stateful to stateless reduction, it does leave a number of open questions:
\begin{enumerate}
\item \textbf{Security against polynomially many queries.} Can our security proof be strengthened to show information theoretic security against a polynomial number of queries to the token? We conjecture this to be the case, but finding a formal proof has been elusive. (See discussion under ``Towards security against polynomially many adaptive queries'' above for details.)

\item \textbf{Composable security against a malicious sender.} While we show composable security against a malicious receiver, our protocol can achieve standalone  security against a malicious sender.  Could an adaptation of our protocol ensure composable security against a malicious sender as well?\footnote{We note that this would require a different protocol, since in our current construction, a cheating sender could program the token to abort based on the user's input.}

\item \textbf{Non-reversible token.} Our impossibility result for quantum one-time memories with \emph{quantum} queries (Section~\ref{sec:Impossibility}) assumes the adversary has access to reversible tokens; can a similar impossibility result be shown for non-reversible tokens? In Section~\ref{sec:Impossibility}, we briefly discuss why it may be difficult to extend the techniques of our impossibility results straightforwardly when the adversary does \emph{not} have access to the inverse of the token.

\item \textbf{Imperfect devices.} While our prepare-and-measure scheme is technologically simple, it is still virtually unrealizable with current technology, due to the requirement of perfect quantum measurements. We leave open the question of tolerance to a small amount of noise.
\end{enumerate}

\medskip

\noindent{\bf Organization.} We begin in Section~\ref{scn:prelims} with preliminaries, including the ideal functionalities for an OTM and stateless token, background on quantum channels, semidefinite programming, and the Gutoski-Watrous framework for quantum games.
In Section~\ref{scn:feasiblity}, we give our construction for an OTM based on a stateless hardware token; the proof ideas for security are also provided.
In Section~\ref{sec:Impossibility}, we discuss ``tightness'' of our construction by showing two impossibility results for ``relaxations'' of our scheme.
In the Appendix, we include the description of classical UC and quantum UC  (Appendix~\ref{appendix:UCmodels}); Appendix~\ref{sec:appendix-def-malicious-sender} establishes notation required in the definition of stand-alone security against a malicious sender.
Appendix~\ref{scn:newproof} gives our formal security proof against a linear number of queries to the token; these results are used to finish the security proof in Section~\ref{scn:feasiblity}. Appendix~\ref{app:cleaner} gives a simplification of the GW SDP, derives its dual, and gives a dual feasible solution which we conjecture to be approximately optimal (formally stated in Conjecture~\ref{conj:only}). Finally, the security proof for a lemma in Section~\ref{sec:Impossibility} can be found in Appendix~\ref{app:4.1}.

\section{Preliminaries}\label{scn:prelims}

\noindent{\bf Notation. }
Two binary distributions $\mathbf{X} $ and $\mathbf{Y}$ are {\em indistinguishable}, denoted
$\mathbf{X} \approx \mathbf{Y}$, if
\begin{equation}
    \left|\Pr(X_n = 1) - \Pr(Y_n =1)\right| \leq \text{negl}(n).
\end{equation}
We define single-qubit $\ket{0}_+ = \ket{0}$ and  $\ket{1}_+ = \ket{1}$, so that $\{\ket{0}_+, \ket{1}_+\}$ form the \emph{rectilinear basis}. We define
 $\ket{0}_\times = \frac{1}{\sqrt{2}}(\ket{0} + \ket{1})$ and $\ket{1}_\times = \frac{1}{\sqrt{2}}(\ket{0} - \ket{1})$, so that $\{\ket{0}_\times, \ket{1}_\times\}$ form the \emph{diagonal basis}.
For strings  $x= x_1,x_2, \ldots x_n \in \{0,1\}^n$ and $\theta = \theta_1, \theta_2, \ldots, \theta_n \in \{+, \times\}^n$, define $\ket{x}_\theta = \bigotimes_{i=1}^n\ket{x_i}_{\theta_i}$. The Hadamard gate in quantum information is $H=(1~1;~1~-1)/\sqrt{2}$. It maps $H\ket{0}_+=\ket{0}_\times$, $H\ket{1}_+=\ket{1}_\times$, $H\ket{0}_\times=\ket{0}_+$, and $H\ket{1}_\times=\ket{1}_+$. For $\spa{X}$ a finite dimensional complex Hilbert space, $\mathcal{L}(\spa{X})$, $\operatorname{Herm}(\spa{X})$, $\operatorname{Pos}(\spa{X})$, and $\spa{D}(\spa{X})$ denote the sets of linear, Hermitian, positive semidefinite, and density operators acting on $\spa{X}$, respectively. The notation $A\succeq B$ means $A-B$ is positive semidefinite.

\medskip

\noindent{\bf Quantum universal composition (UC) framework.}\
We consider simulation-based security in this paper. In particular, we prove the security of our construction against a malicious receiver in the quantum universal composition (UC) framework~\cite{EC:Unruh10}. Please see Appendix~\ref{appendix:UCmodels} for a brief description of the classical UC~\cite{FOCS:Canetti01} and the quantum UC~\cite{EC:Unruh10}.  In the next two paragraphs, we introduce two relevant ideal functionalities of one-time memory and of stateless hardware token.

\medskip

\noindent{\bf One-time memory (OTM).}\
The  one-time memory (OTM) functionality $\funcOTM$ involves two parties, the sender and the receiver, and consists of two phases, ``Create'' and ``Execute''.
Please see {Functionality~\ref{ideal-funct:OTM}} below for details; for the sake of simplicity, we have omitted the session/party identifiers as they should be implicitly clear from the context.
 We sometimes refer to this functionality
$\funcOTM$ as an \emph{OTM token}.

\begin{functionality}[htbp!]
\caption{Ideal functionality $\funcOTM$.
\label{ideal-funct:OTM}}
\begin{enumerate}
\item \textbf{Create:} Upon input~$(s_0, s_1)$ from the sender, with $s_0, s_1 \in \{0,1\}$, send \create to the
receiver and store $(s_0, s_1)$.
\item \textbf{Execute:} Upon input $b\in \{0,1\}$  from the
receiver, send $s_b$ to  receiver. Delete any trace of this instance.
\end{enumerate}
\end{functionality}

\medskip

\noindent{\bf Stateless hardware.} \
The original work of Katz~\cite{EC:Katz07} introduces the ideal functionality $\funcHT$ to model stateful tokens in the UC-framework.
 In the ideal model, a party that wants to create a token, sends a Turing machine to $\funcHT$.
$\funcHT$ will then run the machine
(keeping the state), when the designated party will ask for it. The same functionality can be adapted
to model stateless tokens. It is sufficient that the functionality does not keep the state between two
executions.
A simplified version of the $\funcHT$ functionality as shown in~\cite{EC:ChaGoySah08} (that is very similar to the $\funcHT$
of~\cite{EC:Katz07}) is described below.
Note that, again for the sake of simplicity, we have omitted the session/party identifiers as they should be implicitly clear from
the context.

\begin{functionality}[htbp!]
\caption{ Ideal functionality $\funcHT$.
\label{ideal-funct:HT}}
The functionality is
parameterized by a polynomial $p(\cdot)$, and an implicit security parameter $\secp$.
\begin{enumerate}
\item  \textbf{Create:} Upon input $(\create, M)$ from the sender, where $M$ is a Turing machine,  send \create to the
receiver and store $M$.
\item \textbf{Execute:} Upon input $(\run, msg)$  from the
receiver, execute  $ M(msg)$ for at most $p(\secp)$ steps, and let $out$ be the response. Let $out:=\bot$ if $M$ does not halt in $p(\secp)$ steps.  Send $out$ to the receiver.
\end{enumerate}
\end{functionality}

Although the environment and adversary are unbounded, we specify that stateless hardware can be queried only a polynomial number of times. This is necessary; otherwise the hardware token model is vacuous (with unbounded queries, the entire input-output behavior of stateless hardware can be extracted).

\paragraph{Quantum channels.} We now review quantum channels. A basic background in quantum information is assumed, see e.g.~\cite{NC00} for a standard reference. A linear map $\Phi:\spa{L}(\spa{X})\mapsto \spa{L}(\spa{Y})$ is a \emph{quantum channel} if $\Phi$ is trace-preserving and completely positive (TPCP). Such maps take density operators to density operators. A useful representation of linear maps (or ``superoperators'') $\Phi:\spa{L}(\spa{X})\mapsto \spa{L}(\spa{Y})$ is the Choi-Jamio\l{}kowski representation, $J(\Phi)\in\spa{L}(\spa{Y}\otimes \spa{X})$. The latter is defined (with respect to some choice of orthonormal basis $\set{\ket{i}}$ for $\spa{X}$) as
$
    J(\Phi)=\sum_{i,j}\Phi(\ketbra{i}{j})\otimes\ketbra{i}{j}.
$
The following properties of $J(\Phi)$ hold~\cite{C75,J72}: (1) $\Phi$ is completely positive if and only if $J(\Phi)\succeq 0$, and (2) $\Phi$ is trace-preserving if and only if $\trace_{\spa{Y}}(J(\Phi))=I_{\spa{X}}$. In a nutshell, the Gutoski-Watrous (GW) framework generalizes this definition to \emph{interacting} strategies~\cite{GutoskiW07}.

\paragraph{Semidefinite programs.} We give a brief overview of semidefinite programs (SDPs) from the perspective of quantum information, as done \emph{e.g.}, in the notes of Watrous~\cite{W11_7} or~\cite{MVW13}. For further details, a standard text on convex optimization is Boyd and Vandenberghe~\cite{BV04}.

Given any $3$-tuple $(A,B,\Phi)$ for operators $A\in\operatorname{Herm}(\spa{X})$ and $B\in\operatorname{Herm}(\spa{Y})$, and Hermiticity-preseving linear map $\Phi:\spa{L}(\spa{X})\mapsto\spa{L}(\spa{Y})$, one can state a \emph{primal} and \emph{dual} semidefinite program:
    \begin{center}
  \begin{minipage}{2in}
    \centerline{\underline{Primal problem (P)}}\vspace{-7mm}
    \begin{align*}
        \sup \quad&\trace(AX)\\
        \text{s.t.} \quad&\Phi(X)=B,\\
         & X\in\operatorname{Pos}(\spa{X}),
		\end{align*}
  \end{minipage}
  \hspace*{5mm}
  \begin{minipage}{2in}
    \centerline{\hspace{-10mm}\underline{Dual problem (D)}}\vspace{-7mm}
		\begin{align*}
        \inf \quad&\trace(BY)\\
        \text{s.t.} \quad&\Phi^\ast(Y)\succeq A\\
         & Y\in\operatorname{Herm}(\spa{Y}),
  	\end{align*}	
  \end{minipage}
\end{center}
where $\Phi^*$ denotes the \emph{adjoint} of $\Phi$, which is the unique map satisfying $
\trace(A^\dagger\Phi(B))=\trace((\Phi^\ast (A))^\dagger B)
$
for all $A\in\spa{L}(\spa{Y})$ and $B\in\spa{L}(\spa{X})$. Not all SDPs have feasible solutions (\emph{i.e.} a solution satisfying all constraints); in this case, we label the optimal values as $-\infty$ for P and $\infty$ for D, respectively. Note also that the SDP we derive in Equation~(\ref{eqn:11}) will for simplicity not be written in precisely the form above, but can without loss of generality be made so.

\subsection{The Gutoski-Watrous framework for quantum games}\label{sscn:gw}

We now recall the Gutoski-Watrous (GW) framework for quantum games~\cite{GutoskiW07}, which can be used to model quantum interactions between spatially separated parties. The setup most relevant to our protocol here is depicted in Figure~\ref{fig:interaction}.
    \begin{figure}[t]
    \centering
      \includegraphics[width=13cm]{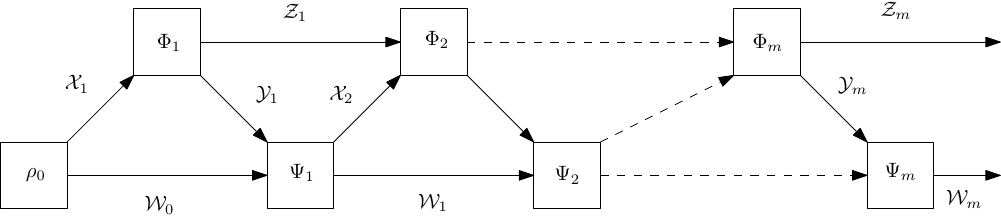}
      \caption{A general interaction between two quantum parties.}
      \label{fig:interaction}
    \end{figure}
    Here, we imagine one party, $A$, prepares an initial state $\rho_0\in\dens(\spa{X}_1\otimes \spa{W}_0)$. Register $\spa{X}_1$ is then sent to the second party ($\spa{W}_0$ is kept as private memory), $B$, who applies some quantum channel $\Phi_i:\mathcal{L}(\spa{X}_1)\mapsto\mathcal{L}(\spa{Y}_1\otimes\spa{Z}_1)$. $B$ keeps register $\spa{Z}_1$ as private memory, and sends $\spa{Y}_1$ back to $A$, who applies channel $\Psi_1:\mathcal{L}(\spa{W}_0\otimes\spa{Y}_1)\mapsto\mathcal{L}(\spa{X}_2\otimes \spa{W}_1)$, and sends $\spa{X}_2$ to $B$. The protocol continues for $m$ messages back and forth, until the final operation $\Psi_m:\mathcal{L}(\spa{W}_m\otimes\spa{Y}_m)\mapsto\complex$, in which $A$ performs a two-outcome measurement (specifically, a POVM $\Lambda=\set{\Lambda_0, \Lambda_1}$, meaning $\Lambda_0,\Lambda_1\succeq 0$, $\Lambda_0+\Lambda_1=I$) in order to decide whether to reject ($\Lambda_0$) or accept ($\Lambda_1$). As done in~\cite{GutoskiW07}, we may assume without loss of generality\footnote{This is due to the Stinespring dilation theorem.} that all channels are given by linear isometries\footnote{A linear isometry $A\in\mathcal{L}(\spa{S},\spa{T})$ satisfies $A^\dagger A=I_{\spa{S}}$,  generalizing the notion of unitary maps to non-square matrices.} $A_k$, i.e. $\Phi_k(X)=A_k X A_k^\dagger$. Reference \cite{GutoskiW07} refers to $(\Phi_1,\ldots,\Phi_m)$ as a \emph{strategy} and $(\rho_0,\Psi_1,\ldots, \Psi_m)$ as a \emph{co-strategy}.  In our setting, the former is ``non-measuring'', meaning it makes no final measurement after $\Phi_m$ is applied, whereas the latter is ``measuring'', since we will apply a final measurement on space $\spa{W}_m$ (not depicted in Figure~\ref{fig:interaction}).

    Intuitively, since our protocol (Section~\ref{section:qOTM}) will begin with the token sending the user a quantum key $\ket{x}_\theta$, we will later model the token as a \emph{measuring co-strategy}, and the user of the token as a \emph{strategy}. The advantage to doing so is that the GW framework allows one to (recursively) characterize any such strategy (resp., co-strategy) via a set of linear (in)equalities and positive semi-definite constraints. (In this sense, the GW framework generalizes the {Choi-Jamio\l{}kowski} representation for channels to a ``{Choi-Jamio\l{}kowski}'' representation for strategies/co-strategies.) To state these constraints, we first write down the {Choi-Jamio\l{}kowski} (CJ) representation of a strategy (resp., measuring co-strategy) from \cite{GutoskiW07}.

    \paragraph{CJ representation of (non-measuring) strategy.} The CJ representation of a strategy $(A_1,\ldots, A_m)$ is given by matrix~\cite{GutoskiW07}
    \begin{equation}\label{eqn:pt}
        \trace_{\spa{Z}_m}(\operatorname{vec}(A)\operatorname{vec}(A)^\dagger),
    \end{equation}
    where $A\in\mathcal{L}(\spa{X}_1\otimes\cdots\otimes \spa{X}_m,\spa{Y}_1\otimes\cdots\otimes\spa{Y}_m\otimes \spa{Z}_m)$ is defined as the product of the isometries $A_i$,
    \begin{equation}\label{eqn:A}
        A:=(I_{\spa{Y}_1\otimes\cdots\otimes\spa{Y}_{m-1}}\otimes A_m)\cdots(A_1\otimes I_{\spa{X}_2\otimes\cdots\otimes\spa{X}_m}),
    \end{equation}
    and the $\operatorname{vec}:\spa{L}(\spa{S},\spa{T})\mapsto\spa{T}\otimes\spa{S}$ mapping is the linear extension of the map $\ketbra{i}{j}\mapsto\ket{i}\ket{j}$ defined on all standard basis states $\ket{i},\ket{j}$.

    \paragraph{CJ representation of (measuring) co-strategy.} Let $\Lambda:=\set{\Lambda_0,\Lambda_1}$ denote a POVM with reject and accept measurement operators $\Lambda_0$ and $\Lambda_1$, respectively. A measuring strategy which ends with a measurement with respect to POVM $\Lambda$ replaces, for $\Lambda_a\in \Lambda$, Equation~(\ref{eqn:pt}) with~\cite{GutoskiW07}
    \begin{align}
        Q_a&:=\trace_{\spa{Z}_m}((\Lambda_a\otimes I_{\spa{Y}_1\otimes\cdots\otimes \spa{Y}_m})\operatorname{vec}(A)\operatorname{vec}(A)^\dagger)\\
        &=\trace_{\spa{Z}_m}(\operatorname{vec}((\sqrt{\Lambda_a}\otimes I_{\spa{Y}_1\otimes\cdots\otimes \spa{Y}_m})A)\operatorname{vec}((\sqrt{\Lambda_a}\otimes I_{\spa{Y}_1\otimes\cdots\otimes \spa{Y}_m})A)^\dagger)\\
        &=:\trace_{\spa{Z}_m}(\operatorname{vec}(B_a)\operatorname{vec}(B_a)^\dagger)\label{eqn:CJ}.
    \end{align}
    To convert this to a \emph{co}-strategy, one takes the transpose of the operators defined above (with respect to the standard basis). (Note: In our use of the GW framework in Section~\ref{sscn:linearsecurity}, all operators we derive will be symmetric with respect to the standard basis, and hence taking this transpose will be unnecessary.)

    \paragraph{Optimization characterization over strategies and co-strategies.} With CJ representations for strategies and co-strategies in hand, one can formulate~\cite{GutoskiW07} the optimal probability with which a strategy can force a corresponding co-strategy to output a desired result as follows. Fix any $Q_a$ from a measuring co-strategy $\set{Q_0,Q_1}$, as in Equation~(\ref{eqn:CJ}). Then, Corollary 7 and Theorem 9 of~\cite{GutoskiW07} show that the maximum probability with which a (non-measuring) strategy can force the co-strategy to output result $a$ is given by
		\begin{align}
			\text{min:}\quad & p\label{eqn:final1}\\
  		\text{subject to:}\quad & Q_a \preceq pR_m\label{eqn:prec}\\
  		& R_k= P_k\otimes I_{\spa{Y}_k}\qquad\qquad\qquad &\text{for }1\leq k\leq m\label{eqn:cond1}\\
        & \trace_{\spa{X}_k}(P_k)= R_{k-1} &\text{for }1\leq k\leq m\label{eqn:cond2}\\
        & R_0 = 1\\
        & R_k\in \operatorname{Pos}(\spa{Y}_{1,\ldots, k}\otimes \spa{X}_{1,\ldots, k})&\text{for }1\leq k\leq m\\
        & P_k\in \operatorname{Pos}(\spa{Y}_{1,\ldots, k-1}\otimes \spa{X}_{1,\ldots, k})&\text{for }1\leq k\leq m\label{eqn:lastcond}\\
        &p\in[0,1] &
          	\end{align}	
        \paragraph{Intuition.} The minimum $p$ denotes the optimal ``success'' probability, meaning the optimal probability of forcing the co-strategy to output $a$ (Theorem 9 of~\cite{GutoskiW07}). The variables above, in addition to $p$, are $\set{R_i}$ and $\set{P_i}$, where the optimization is happening over all $m$-round co-strategies $R_m$ satisfying Equation~(\ref{eqn:prec}). How do we enforce that $R_m$ encodes such an $m$-round co-strategy? This is given by the (recursive) Equations (\ref{eqn:cond1})-(\ref{eqn:lastcond}). Specifically, Corollary 7 of~\cite{GutoskiW07} states that $R_m$ is a valid $m$-round co-strategy if and only if all of the following hold: (1) $R_m\succeq 0$, (2) $R_m=P_m\otimes I_{\spa{Y}_m}$ for $P_m\succeq 0$ and $\spa{Y}_m$ the last incoming message register to the co-strategy, (3) $\trace_{\spa{X}_m}(P_m)$ is a valid $m-1$ round co-strategy (this is the recursive part of the definition). An intuitive sense as to why conditions (2) and (3) should hold is as follows: For any $m$-round co-strategy $R_m$, let $R_{m-1}$ denote $R_m$ restricted to the first $m-1$ rounds. Then, to operationally obtain $R_{m-1}$ from $R_m$, the co-strategy first ignores the last incoming message in register $\spa{Y}_m$. This is formalized via a partial trace over $\spa{Y}_m$, which (once pushed through the CJ formalism\footnote{Recall that the CJ representation of the trace map is the identity matrix (up to scaling).}) translates into the $\otimes I_{\spa{Y}_k}$ term in Equation~(\ref{eqn:cond1}). Since the co-strategy is now ignoring the last \emph{incoming} message $\spa{Y}_m$, any measurement it makes after $m-1$ rounds is independent of the last \emph{outgoing} message $\spa{X}_m$. Thus, we can trace out $\spa{X}_m$ as well, obtaining a co-strategy $R_{m-1}$ on just the first $m-1$ rounds; this is captured by Equation~(\ref{eqn:cond2}).

\section{Feasibility of Quantum OTMs using Stateless Hardware}\label{scn:feasiblity}

In this section, we present a {\em quantum} construction for one-time memories  by using stateless hardware (Section~\ref{section:qOTM}). We also state our main theorem (Theorem~\ref{thm:main}). In Section~\ref{sscn:proof}, we describe the Simulator and prove Theorem~\ref{thm:main} using the technical results of Appendix~\ref{scn:newproof}. The intuition and techniques behind the proofs in Appendix~\ref{scn:newproof} are sketched in Section~\ref{sscn:securityintuition}.

\subsection{Construction}
\label{section:qOTM}

We now present the OTM protocol $\Pi$ in the $\funcHT$ hybrid model, between a sender $\PS$ and a receiver~$\PR$.
Here the security parameter is $\secp$.

\begin{tiret}
\item Upon receiving input $(s_0,s_1)$ from the environment where  $s_0, s_1 \in \{0,1\}$,  sender $\PS$ acts as follows:

\begin{itemize}
\item The sender chooses uniformly random $x \in_R \{0,1\}^n$ and $\theta \in_R \{+, \times\}^n$, and prepares~$\ket{x}_\theta$.
 Based on tuple $(s_0,s_1,x,\theta)$, the sender then prepares the program $M$ as in \textbf{Program~\ref{hardware-token-program}}.

\begin{program} \caption{ Program for hardware token} \label{hardware-token-program}
 Hardcoded values:  $s_0, s_1 \in \{0,1\}$, $x \in \{0,1\}^n$, and $\theta \in \{+, \times\}^n$\\
Inputs: $y\in\set{0,1}^n$ and $b\in\set{0,1}$,
where $y$ is a claimed measured value for the quantum register, and $b$ the evaluator's choice bit
\begin{enumerate}
\item  If $b=0$, check that the $\theta=+$ positions return the correct bits in $y$ according to~$x$.  If Accept, output $s_0$. Otherwise output~$\bot$.
\item  If $b=1$, check that the $\theta=\times$ positions return the correct bits in $y$ according to~$x$.  If Accept, output $s_1$. Otherwise output~$\bot$.
\end{enumerate}
\end{program}

\item The sender sends $\ket{x}_\theta$ to the receiver.
\item The sender sends $(\create, M)$ to functionality $\funcHT$, and the functionality sends  \create to notify the receiver.

\end{itemize}

\item
The receiver $\PR$ operates as follows:\\
Upon  input $b$ from the environment, and $\ket{x}_\theta$ from the receiver, and  \create notification from $\funcHT$,
\begin{itemize}
\item If $b=0$, measure $\ket{x}_\theta$ in the computational basis to get string~$y$. Input $(\run, (y,b))$ into $\funcHT$.
\item If $b=1$, apply $\hgate^{\otimes n}$ to $\ket{x}_\theta$, then measure in the computational basis to get string~$y$. Input $(\run, (y,b))$ into $\funcHT$.
\end{itemize}
Return the output of $\funcHT$ to the environment. \\
 It is easy to see that the output of $\funcHT$ is $s_b$ for both $b=0$ and $b=1$.
\end{tiret}
Note again that the hardware token, as defined in \textbf{Program~\ref{hardware-token-program}}, accepts only classical input (\emph{i.e.}, it cannot be queried in superposition). As mentioned earlier, relaxing this assumption yields impossibility of a secure OTM implementation (assuming the receiver also has access to the token's inverse operation), as shown in Section~\ref{sec:Impossibility}.

%
\subsection{Stand-Alone Security Against a Malicious Sender}
\label{sec:security-sender}

We note that in protocol $\Pi$ of Section \ref{section:qOTM},  once the sender prepares and sends the token, she is no longer involved (and in particular, the sender does not receive any further communication from the receiver). We call such a protocol a \emph{one-way} protocol.
Because of this simple structure, and because the ideal functionality $\funcHT$ also does not return any message to the sender, we can easily establish stand-alone security against a malicious sender (see details in Appendix~\ref{sec:appendix-def-malicious-sender}).

\subsection{UC-Security against a corrupt receiver}
\label{sscn:proof}

Our main theorem, which establishes security against a corrupt receiver is now stated as follows.

\begin{theorem}\label{thm:main}
Construction $\Pi$ above quantum-UC-realizes $\funcOTM$ in the $\funcHT$ hybrid model  with statistical security against an  actively-corrupted receiver making at most $cn$ number of adaptive queries to the token, for any fixed constant $c<0.114$.
\end{theorem}

\noindent To prove Theorem~\ref{thm:main}, we must construct and analyze an appropriate simulator, which we now do.

\subsubsection{The simulator}
\label{sec:simul}

In order to prove Theorem~\ref{thm:main}, for an adversary $\cA$ that corrupts the receiver, we build a simulator $\cS$ (having access to the OTM functionality $\funcOTM$), such that for any unbounded environment $\cZ$, the executions in the real model and that in simulation are statistically indistinguishable.
Our simulator $\cS$ is given below:

\begin{tiret}

\item The simulator emulates an internal copy of the adversary $\cA$ who corrupts the receiver. The simulator emulates the communication between $\cA$ and the external environment $\cZ$ by forwarding the communication messages between $\cA$ and $\cZ$.
\item The simulator $\cS$ needs to emulate the whole view for the adversary~$\cA$.
First, the simulator picks dummy inputs $\tilde{s}_0 =0$ and  $\tilde{s}_1 =0$, and randomly chooses $x \in \{0,1\}^n$, and $\theta \in \{+, \times\}^n$, and generates program~$\tilde M$.
Then the simulator plays the role of the sender to send $\ket{x}_\theta$ to the adversary $\cA$ (who controls the corrupted receiver). The simulator also emulates  $\funcHT$ to notify~$\cA$ by sending \create to indicate that the hardware is ready for queries.

\item  For each query $(\run,(b,y))$ to $\funcHT$ from the adversary $\cA$, the simulator evaluates program $\tilde M$ (that is created based on $\tilde s_0, \tilde s_1, x, \theta$) as in the construction, and then
 acts as follows:
\begin{enumerate}
\item If this is a rejecting input, output $\bot$.
\item If this is the first accepting input, call the external $\funcOTM$ with input~$b$, and learn the output  $s_b$ from $\funcOTM$. Output~$s_b$.
\item If this is a subsequent accepting input, output $s_b$ (as above).
\end{enumerate}
\end{tiret}

\subsubsection{Analysis}\label{sscn:analysis}
We now show  that the simulation and the real model execution are statistically indistinguishable.
There are two cases in an execution of
 the simulation which we must consider:

\begin{tiret}
\item \emph{Case 1: In all its queries to $\funcHT$,  the accepting inputs of $\cA$ have the same choice bit~$b$.} In this case, the simulation is perfectly indistinguishable.
\item  \label{case:2} \emph{Case 2: In its queries to $\funcHT$,  $\cA$ produces accepting inputs for both $b=0$ and $b=1$.} In this case,
    it is possible that
    the simulation  fails (the environment can distinguish the real model from the ideal model), since the simulator is only able to retrieve a single bit from the external OTM functionality $\funcOTM$ (either corresponding to $b=0$ or $b=1$).
\end{tiret}
Thus, whereas in Case 1 the simulator behaves perfectly, in Case 2 it is in trouble. Fortunately, in Theorem~\ref{thm:securetoken} we show that the probability that Case~2 occurs is exponentially small in $n$, the number of qubits comprising $\ket{x}_\theta$, provided the number of queries to the token is at most~$cn$ for any $c<0.114$. Specifically, we show that for an arbitrary $m$-query strategy (\emph{i.e.}, any quantum strategy allowed by quantum mechanics, whether efficiently implementable or not, which queries the token at most $m$ times), the probability of Case~$2$ occurring is at most $O(2^{2m-0.228n})$.  This concludes the proof. \looseness=-1

\subsection{Security analysis for the token: Intuition}\label{sscn:securityintuition}

Our simulation proof showing statistical security of our Quantum OTM construction of Section~\ref{section:qOTM} relies crucially on  Theorem~\ref{thm:securetoken}, stated below. As the proof of this theorem uses quantum information theoretic and semidefinite programming techniques (as opposed to cryptographic techniques), let us introduce notation in line with the formal analysis of Appendix~\ref{scn:newproof}.

With respect to the construction of Section~\ref{section:qOTM}, let us replace each two-tuple $(x,\theta)\in\set{0,1}^n\times\set{+,\times}^n$ by a single string $z\in\set{0,1}^{2n}$, which we denote the \emph{secret key}. Bits $2i$ and $2i+1$ of $z$ specify the basis and value of conjugate coding qubit $i$ for $i\in\set{1,\ldots, n}$ (\emph{i.e.}, $z_{2i}=\theta_{i}$ and $z_{2i+1}=x_i$). Also, rename the ``quantum key'' (or conjugate coding key) $\ket{\psi_z}:=\ket{x}_\theta\in(\complex^2)^{\otimes n}$. Thus, the protocol begins by having the sender pick a \emph{secret key} $z\in\set{0,1}^{2n}$ uniformly at random, and preparing a joint state
\begin{equation}\label{eqn:rhoR}
    \ket{\psi}=\frac{1}{2^n}\sum_{z\in\ket{0,1}^{2n}}\ket{\psi_z}_R\ket{z}_T.
\end{equation}
The first register, $R$, is sent to the receiver, while the second register, $T$, is kept by the token. (Thus, the token knows the secret key $z$, and hence also which $\ket{\psi_z}$ the receiver possesses.) The mixed state describing the receiver's state of knowledge at this point is given by
\begin{equation}
    \rho_R := \frac{1}{2^{2n}}\sum_{z\in\set{0,1}^{2n}}\ketbra{\psi_z}{\psi_z}.
\end{equation}

\begin{theorem}\label{thm:securetoken}
    Given a single copy of $\rho_R$, and the ability to make $m$ (adaptive) queries to the hardware token, the probability that an unbounded quantum adversary can force the token to output both bits $s_0$ and $s_1$ scales as $O(2^{2m-0.228n})$.
\end{theorem}
\noindent Thus, the probability of an unbounded adversary (\emph{i.e.}, with the ability to apply the most general maps allowed in quantum mechanics, trace-preserving completely positive (TPCP) maps, which are not necessarily efficiently implementable) to successfully cheat using $m=cn$ for $c<0.114$ queries is exponentially small in the quantum key size, $n$.
The proof of Theorem~\ref{thm:securetoken} is in Appendix~\ref{scn:newproof}. Its intuition can be sketched as follows.

\paragraph{Proof intuition.} The challenge in analyzing security of the protocol is the fact that the receiver (a.k.a. the user) is not only given adaptive query access to the token, but also a copy of the quantum ``resource state'' $\rho_R$, which it may arbitrarily tamper with (in any manner allowed by quantum mechanics) while making queries. Luckily, the GW framework~\cite{GutoskiW07} (Section~\ref{sscn:gw})) is general enough to model such ``queries with quantum side information''. The framework outputs an SDP, $\Gamma$ (Equation~(\ref{eqn:1})), the optimal value of which will encode the optimal cheating probability for a cheating user of our protocol. Giving a feasible solution for $\Gamma$ will hence suffice to upper bound this cheating probability, yielding Theorem~\ref{thm:securetoken}.\\

\noindent \emph{Coherently modeling quantum queries to the token.} To model the interaction between the token and user, we first recall that all queries to the token must be classical by assumption. To model this process \emph{coherently} in the GW framework, we hence imagine (solely for the purposes of the security analysis) that the token behaves as follows:
\begin{enumerate}
    \item It first sends state $\rho_R$ to the user.
    \item When it receives as $i$th query a quantum state $\rho_i$ from the user, it sends response string $r_i$ to the user, and ``copies'' $\rho_i$ via transversal CNOT gates to a private memory register $\spa{W}_i$, along with $r_i$. It does not access $\rho_i$ again throughout the protocol, and only accesses $r_i$ again in Step 3. For clarity, the token runs a classical circuit, and in the formal setup of Appendix~\ref{scn:newproof} (see Remark~(\ref{rem})), the token conditions each response $r_i$ solely on the current incoming message, $\rho_i$.
    \item After all rounds of communication, the token ``measures'' its stored responses $(r_1,\ldots, r_m)$ in the $Z$-basis to decide whether to accept (user successfully cheated\footnote{We model the token as ``accepting'' when the user successfully cheats, so that a feasible solution to the semidefinite program $\Gamma$ of Equation~(\ref{eqn:1}) correctly \emph{upper bounds} the probability of said cheating. Formally, in the GW framework (Section~\ref{sscn:gw}), we will let $\Lambda_1$ denote this accepting measurement for the token; see Appendix~\ref{scn:newproof}.}) or reject (user failed to cheat).
\end{enumerate}
The ``copying'' phase of Step 2 accomplishes two tasks: First, since the token will never read the ``copies'' of $\rho_i$ again, the principle of deferred measurement~\cite{NC00} implies the transversal CNOT gates effectively simulate measuring $\rho_i$ in the standard basis. In other words, without loss of generality the user is reduced to feeding a classical string $\widetilde{y}$ to the token. Second, we would like the entire security analysis to be done in a unified fashion in a single framework, the GW framework. To this end, we want the token itself to ``decide'' at the end of the protocol whether the user has successfully cheated (i.e. extracted both secret bits). Storing all responses $r_i$ in Step 2 allows us to simulate such a final measurement in Step 3. We reiterate that, crucially, once the token ``copies'' $\rho_i$ and $r_i$ to $W_i$, it (1) never accesses (i.e. reads or writes to) $\rho_i$ again and (2) only accesses $r_i$ again in the final standard basis measurement of Step $3$. Together, these ensure all responses $r_i$ are independent, as required for a stateless token. A more formal justification is in Remark~\ref{rem} of Appendix~\ref{scn:newproof}.\\

\noindent \emph{Formalization in GW framework.} To place the discussion thus far into the formal GW framework, we return to Figure~\ref{fig:interaction}. The bottom ``row'' of Figure~\ref{fig:interaction} will depict the token's actions, and the top row the user's actions. As outlined above, the protocol begins by imagining the token sends initial state $\rho_0=\rho_R$ to the user via register $\spa{X}_1$. The user then applies an arbitrary sequence of TPCP maps $\Phi_i$ to its private memory (modeled by register $\spa{Z}_i$ in round $i$), each time sending a query $\widetilde{y}_i$ (which is, as discussed above a classical string without loss of generality) to the token via register $\spa{Y}_i$. Given any such query $\widetilde{y}_i$ in round $i$, the token applies its own TPCP map $\Psi_i$ to determine how to respond to the query. In our protocol, the $\Psi_i$ correspond to coherently applying a classical circuit, i.e. a sequence of unitary gates mapping the standard basis to itself. Specifically, their action is fully determined by Program~\ref{hardware-token-program}, and {in principle} all $\Psi_i$ are identical since the token is stateless (\emph{i.e.}, the action of the token in round $i$ is unaffected by previous rounds $\set{1,\ldots, i-1}$). (We use the term ``in principle'', as recall from above that in the security analysis we model each $\Psi_i$ as classically copying $(\widetilde{y}_i,r_i)$ to a distinct private register $W_i$.) Finally, after receiving the $m$th query $\widetilde{y}_m$ in register $\spa{Y}_m$, we imagine the token makes a measurement (not depicted in Fig.~\ref{fig:interaction}) based on the query responses $(r_1,\ldots, r_m)$ it returned; if the user managed to extract both $s_0$ and $s_1$ via queries, then the token ``accepts''; otherwise it ``rejects''. (Again, we are using the fact that in our security analysis, the token keeps a history of all its responses $r_i$, solely for the sake of this final measurement.)

With this high-level setup in place, the output of the GW framework is a semidefinite program\footnote{Note the optimization in Equation~(\ref{eqn:1}) differs from that in Equation~(\ref{eqn:final1}). This is because, technically, Equation~(\ref{eqn:1}) is not yet an SDP due to the quadratic constraint $Q_a\preceq p R_{m}$. It is, however, easily seen to be equivalent to the SDP in Equation~(\ref{eqn:final1}). We thank Jamie Sikora for pointing this out to us.}, denoted $\Gamma$ (see Appendix~\ref{scn:newproof} for further details):
		\begin{align}
			\text{min:}\quad & p\label{eqn:1}\\
  		\text{subject to:}\quad & Q_1 \preceq R_{m+1}\label{eqn:2}\\
  		& R_k= P_k\otimes I_{\spa{Y}_k}\qquad\qquad\qquad &\text{for }1\leq k\leq m+1\label{eqn:3}\\
        & \trace_{\spa{X}_k}(P_k)= R_{k-1} &\text{for }1\leq k\leq m+1\label{eqn:4}\\
        & R_0 = p\label{eqn:5}\\
        & R_k\in \operatorname{Pos}(\spa{Y}_{1,\ldots, k}\otimes \spa{X}_{1,\ldots, k})&\text{for }1\leq k\leq m+1\label{eqn:6}\\
        & P_k\in \operatorname{Pos}(\spa{Y}_{1,\ldots, k-1}\otimes \spa{X}_{1,\ldots, k})&\text{for }1\leq k\leq m+1\label{eqn:7}
          	\end{align}	
 Above, $Q_1$ encodes the actions of the token, i.e. the co-strategy in the bottom row of Figure~\ref{fig:interaction}. The variable $p$ denotes an upper bound on the optimal cheating probability (\emph{i.e.}, the probability with which both $s_0$ and $s_1$ are extracted), subject to linear constraints (Equations~(\ref{eqn:3})-(\ref{eqn:7})) which enforce that operator $R_{m+1}$ encodes a valid co-strategy (see Section~\ref{sscn:gw}). Theorem 9 of~\cite{GutoskiW07} now says that the minimum $p$ above encodes precisely the optimal cheating probability for a user which is constrained only by the laws of quantum mechanics. Since $\Gamma$ is a minimization problem, to upper bound the cheating probability it hence suffices to give a feasible solution $(p,R_1,\ldots,R_{m+1},P_1,\ldots, P_{m+1})$ for $\Gamma$, which will be our approach.

  \paragraph{Intuition for $Q_1$ and an upper bound on $p$.} It remains to give intuition as to how one derives $Q_1$ in $\Gamma$, and how an upper bound on the optimal $p$ is obtained. Without loss of generality, one may assume that each of the token's TPCP maps $\Psi_i$ are given by \emph{isometries} $A_i:\spa{Y}_i\otimes\spa{W}_{i-1}\mapsto\spa{X}_{i+1}\otimes\spa{W}_i$, meaning $A_i^\dagger A_i=I_{\spa{Y}_i\otimes\spa{W}_{i-1}}$ (due to the Stinespring dilation theorem). (We omit the first isometry which prepares state $\rho_0$ in our discussion here for simplicity.) Let us denote their sequential application by a single operator $A:=A_m\cdots A_1$ (note: to make the product well-defined, in Equation~(\ref{eqn:A}) of Appendix~\ref{scn:newproof}, one uses tensor products with identity matrices appropriately). Then, the Choi-Jamio\l{}kowski representation of $A$ is given by~\cite{GutoskiW07} (see Section~\ref{sscn:gw})
  \begin{equation}
    \trace_{\spa{Z}_m}(\operatorname{vec}(A)\operatorname{vec}(A)^\dagger),
  \end{equation}
  where we trace out the token's private memory register $\spa{Z}_m$. (The operator $\operatorname{vec}(\cdot)$ reshapes matrix $A$ into a vector; its precise definition is given in Section~\ref{sscn:gw}.) However, since in our security analysis, we imagine the token also makes a final measurement via some POVM $\Lambda=\set{\Lambda_0,\Lambda_1}$, whereupon obtaining outcome $\Lambda_1$ the token ``accepts'', and upon outcome $\Lambda_0$ the token rejects, we require a slightly more complicated setup. Letting $B_1:= \Lambda_1 A$, we define $Q_1$ as~\cite{GutoskiW07}
  \begin{equation}
  Q_1 = \trace_{\spa{Z}_m}(\operatorname{vec}(B_1)\operatorname{vec}(B_1)^\dagger).
  \end{equation}

  \noindent The full derivation of $Q_1$ in our setting takes a few steps (App.~\ref{scn:newproof}). Here, we state a slightly simplified version of $Q_1$ for exposition with intuition:
  \begin{align}
    Q_1&=\frac{1}{4^{n}}\sum_{\text{``successful''}r}\ketbra{r_m}{r_m}_{\spa{X}_{m+1}}\otimes\cdots\otimes\ketbra{r_1}{r_1}_{\spa{X}_{2}}\otimes\\
        &\hspace{4mm}\left(\sum_{\substack{\text{messages }\widetilde{y}\text{ and keys }z\\\text{consistent with $r$}}}\ketbra{\widetilde{y}_m}{\widetilde{y}_m}_{\spa{Y}_m}\otimes \cdots \otimes\ketbra{\widetilde{y}_1}{\widetilde{y}_1}_{\spa{Y}_{1}}\otimes \ketbra{\psi_z}{\psi_z}_{\spa{X}_1}\right).
    \end{align}
\noindent Above, recall each string $r_i$ denotes the response of the token given the $i$th query $\widetilde{y}_i$ from the user; hence, the corresponding projectors in $Q_1$ act on spaces $\spa{X}_2$ through $\spa{X}_{m+1}$. We say $r$ is ``successful'' if it encodes the user successfully extracting both secret bits from the token. Each string $\widetilde{y}_i\in\set{0,1}^{n+1}$ denotes the $i$th query sent from the user to the token, where each $\widetilde{y}_i=b_i\circ y_i$ in the notation of Program~\ref{hardware-token-program}, \emph{i.e.} $b_i\in \set{0,1}$ is the choice bit for each query. Each such message is passed via register $\spa{Y}_i$. The states $\ket{\psi_z}$ and strings $z$ are defined as in the beginning of Section~\ref{sscn:securityintuition}; recall $z\in\set{0,1}^{2n}$ and $\ket{\psi_z}\in(\complex^2)^{\otimes n}$ denote the secret key and corresponding quantum key, respectively. The inner summation is over all messages $\widetilde{y}$ and keys $z$ such that the token correctly returns response $r_i$ given both $\widetilde{y}_i$ and $z$.

\emph{Upper bounding $p$.} To now upper bound $p$, we give a feasible solution $R_{m+1}$ satisfying the constraints of~$\Gamma$. Note that giving even a solution which attains $p=1$ for all $n$ and $m$ is \emph{non}-trivial --- such a solution is given in Lemma~\ref{l:redundant} of Appendix~\ref{sscn:linearsecurity}. Here, we give a solution which attains $p\in O(2^{2m-0.228n})$, as claimed in Theorem~\ref{thm:securetoken} (and formally proven in Theorem~\ref{thm:cheatingbound} of Appendix~\ref{sscn:linearsecurity}). Namely, we set
\begin{equation}
    R_{m+1}=\frac{1}{N}\sum_{\text{``successful'' }r}\ketbra{r_m}{r_m}_{\spa{X}_{m+1}}\otimes\cdots\otimes\ketbra{r_1}{r_1}_{\spa{X}_{2}}\otimes I_{Y_1\otimes \cdots\otimes Y_m}  \otimes \frac{1}{2^n}I_{\spa{X}_1},
\end{equation}
where intuitively $N$ is the total number of strings $r$ corresponding to successful cheating, and recall $n$ is the key size. This satisfies constraint~(\ref{eqn:3}) of $\Gamma$ due to the identity term $I_{Y_1\otimes \cdots\otimes Y_m}$. The renormalization factor of $(N2^n)^{-1}$ above ensures that tracing out all $\spa{X}_i$ registers yields $R_0=1$ in constraint~(\ref{eqn:5}) of $\Gamma$. We are thus reduced to choosing the minimum $p$ such that constraint $(\ref{eqn:2})$ is satisfied. Note that setting $p=1$ will \emph{not} work for large enough $m$ for this choice of $R_{m+1}$. To see why, observe we have chosen $R_{m+1}$ to align with the block-diagonal structure of $Q_1$ on registers $\spa{X}_{2},\ldots,\spa{X}_m$. Since registers $\spa{Y}_1\otimes \cdots\otimes \spa{Y}_m$ and $\spa{X}_1$ of $R_{m+1}$ are proportional to the identity matrix, it thus suffices to characterize the largest eigenvalue of $Q_1$, $\lmax(Q_1)$. This is done by Lemma~\ref{l:eval} of Appendix~\ref{sscn:linearsecurity}, which says
\begin{equation}
    \lmax(Q_1)=\frac{2}{4^n}\left(1+\frac{1}{\sqrt{2}}\right)^n.
\end{equation}
Combining this bound on $\lmax(Q_1)$ with the parameters of $R_{m+1}$ above now yields the desired claim that $p\in O(2^{2m-0.228n})$. For (say) $m\geq n$ this bound is vacuous, and thus does not suffice to show even the trivial bound $p\leq 1$ for all $m$, as stated. (See Lemma~\ref{l:redundant} for a feasible solution attaining $p\leq 1$ for all $m$.) However, for $m< 0.114 n$ queries, the bound is fruitful, yielding the probability that a user of the token successfully cheats and thus that the simulation fails is exponentially small in the key size,~$n$. Simplifications of the GW SDP, the derivation of its dual SDP, and a conjectured approximately optimal dual feasible solution are given in Appendix~\ref{app:cleaner}.

\section{Impossibility Results}
\label{sec:Impossibility}

We now discuss ``tightness'' of our protocol with respect to impossibility results. To begin, it is easy to argue that OTMs cannot exist in the plain model (\emph{i.e.}, without additional assumptions) in both the classical and quantum settings: in the classical setting, impossibility holds, since software can always be copied. Quantumly, this follows by a simple rewinding argument~\cite{C:BroGutSte13}.
Here, we give two simple no-go results for the quantum setting which support the idea that our scheme is ``tight'' in terms of the minimality of the assumptions it uses. Both results assume the token is reversible, meaning the receiver can run both the token and its inverse operation. The results can be stated as:
\begin{enumerate}
    \item A stateless token which can be queried in \emph{superposition} cannot be used to securely construct an OTM (Section~\ref{sscn:super}).
    \item For \emph{measure and access} schemes such as ours, in order for a stateless token to allow statistical security, it must have an \emph{exponential} number of keys per secret bit (Section~\ref{sscn:bounded}).
\end{enumerate}

Note that if, on the other hand, the receiver is \emph{not} given access to the token's inverse operation, it is unlikely for a straightforward adaption of our no-go techniques to go through. This is because, in the most general case where the token is an arbitrary unitary $U$, which the receiver may apply as a black box, simulating $U^{-1}=U^\dagger$ appears difficult. For example, Theorem 3 of Quintino, Dong, Shimbo, Soeda, and Murao~\cite{QDSSM19} shows that any \emph{exact} implementation of $U^\dagger$ (even with an adaptive protocol) which (1) succeeds with probability $p>0$ and (2) where $p$ is independent of the choice of $U$, requires $k\geq d-1$ uses of $U$. In our setting, $d$ is exponential in the number of qubits, and thus so is $k$. Indeed, inverting arbitrary $U$ would entail, as a special case, inverting arbitrary classical permutations, which appears difficult. For example, Fefferman and Kimmel~\cite{FK18} use precisely this idea (i.e. an in-place permutation oracle, to which one does not have access to the inverse) to prove an oracle separation between two quantum generalizations of NP, Quantum-Classical Merlin Arthur and Quantum Merlin Arthur. We stress, however, that the works of~\cite{QDSSM19,FK18} are for rather general unitaries $U$, whereas here we have a very specific choice of $U$ (i.e. the token's implementation). For such a specialized $U$, it remains possible that a no-go theorem could still hold, even without black-box access to $U^\dagger$.

\subsection{Impossibility: Tokens which can be queried in superposition} \label{sscn:super}

In our construction, we require that all queries to the token be classical strings, \emph{i.e.}, no querying in superposition is allowed. It is easy to argue via a standard rewinding argument that relaxing this requirement yields impossibility of a secure OTM, {as long as access to the token's adjoint (inverse) operation is given}, as we now show. Specifically, let $M$ be a quantum OTM implemented using a hardware token. {Since the token access is assumed to be reversible}, we may model it as an oracle\footnote{This formalization models $O_f$ as a classical function $f$ which can be queried in superposition, since the aim of this paper is to consider ``easy-to-manufacture'' tokens. However, our impossibility arguments in Section~\ref{sec:Impossibility} trivially extend to the case when the token is modelled by an arbitrary unitary $U_f$. } $O_f$ realizing a function $f:\set{0,1}^n\mapsto \set{0,1}^m$ in the standard way, \emph{i.e.}, for all $y\in\set{0,1}^n$ and $b\in\set{0,1}^m$,
$    O_f\ket{y}\ket{b}=\ket{y}\ket{b\oplus f(y)}$. Now, suppose our OTM stores two secret bits $s_0$ and $s_1$, and provides the receiver with an initial state $\ket{\psi}\in A\otimes B\otimes C$, where $A$, $B$, and $C$ are the algorithm's workspace, \emph{query} (\emph{i.e.}, input to $O_f$), and \emph{answer} (\emph{i.e.}, $O_f$'s answers) registers, respectively. By definition, an honest receiver must be able to access precisely one of $s_0$ or $s_1$ with certainty, given $\ket{\psi}$. Thus, for any $i\in\set{0,1}$, there exists a quantum query algorithm $A_i=U_mO_f\cdots O_fU_2O_fU_1$ for unitaries $U_i\in\unitary(A\otimes B\otimes C)$ such that $A_i\ket{\psi}=\ket{\psi'}_{AB}\ket{s_i}_C$. For any choice of $i$, however, this implies a malicious receiver can now classically copy $s_i$ to an external register, and then ``rewind'' by applying $A_i^\dagger$ to $\ket{\psi'}_{AB}\ket{s_i}_C$ to recover $\ket{\psi}$. Applying $A_{i'}$ for $i'\neq i$ to $\ket{\psi}$ now yields the second bit $i'$ with certainty as well. We conclude that a quantum OTM which allows superposition queries to a reversible stateless token is insecure.

\begin{remark}{Above, we assumed the OTM outputs $s_i$ with certainty. The argument generalizes to OTMs that output $s_i$ with probability at least $1-\epsilon$ for small $\epsilon>0$; for this, the Gentle Measurement Lemma~\cite{Winter99} can be used to show that both bits can be recovered with non-negligible probability.
}
\end{remark}

\begin{remark}{
 Our argument crucially relies on the fact that the receiver has superposition access to the $A_i^\dagger$ operation. In certain models (e.g.,~software), such access is unavoidable. However, we do not rule out the possibility that \emph{non-reversible} superposition access to a token would allow for quantum OTMs.
}
\end{remark}

\subsection{Impossibility: Tokens with a bounded number of keys}\label{sscn:bounded} We observed superposition queries to the token prevent an OTM from being secure. One can also ask how simple a hardware token with classical queries can be, while still allowing a secure OTM. Below, we consider such a strengthening in which the token is forced to have a bounded number of keys.

To formalize this, we define the notion of a ``measure-and-access (MA)'' OTM, \emph{i.e.}, an OTM in which given an initial state $\ket{\psi}$, an honest receiver applies a prescribed measurement to $\ket{\psi}$, and feeds the resulting classical string (\emph{i.e.}, key) $y$ into the token $O_f$ to obtain $s_i$. Our construction is an example of a MA memory in which each bit $s_i$ has an \emph{exponential} number of valid keys $y$ such that $f(y)=s_i$. Can the construction can be strengthened such that each $s_i$ has a bounded number (\emph{e.g.}, a polynomial number) of keys? We now show that such a strengthening would preclude security, assuming the token is reversible.

For clarity, implicitly in our proof below, we model the oracle $O_f$ as having three possible outputs: $0$, $1$, or~$2$, where $2$ is output whenever $O_f$ is fed an invalid key $y$. This is required for the notion of having ``few'' keys to make sense (\emph{i.e.},~there are $2^n$ candidate keys, and only two secret bits, each of which is supposed to have a bounded number of keys). Note that our construction indeed fits into this framework.

\begin{lemma}\label{l:impossible}
    Let $M$ be an MA memory with oracle $O_f$, such that $O_f$ cannot be queried in superposition. If a secret bit $s_i$ has at most $\Delta$ keys $y_i$ such that $f(y_i)=s_i$, then given a single copy of $\ket{\psi}$, one can extract both $s_0$ and $s_1$ from $M$ with probability at least $1/\Delta^2$.
\end{lemma}

\noindent We conclude that if a secret bit $b_i$ has (say) at most polynomially many keys, then any measure-and-access OTM can be broken with at least inverse polynomial probability. The proof is given in Appendix~\ref{app:4.1}. In this sense, at least in the paradigm of measure-and-access memories, our construction is essentially tight --- in order to bound the adversary's success probability of obtaining both secret bits by an inverse exponential, we require each secret bit to have exponentially many valid keys. Note that, as in the setting of superposition queries, the above proof can be generalized to the setting in which the OTM returns the correct bit $s_i$ with probability at least $1-\epsilon$ for small $\epsilon>0$. Finally, the question of whether a similar statement to Lemma~\ref{l:impossible} holds for a \emph{non-reversible} token remains open.

\section*{Acknowledgements}
We thank anonymous referees for pointing out that the impossibility result against quantum queries applies only if we model the token as a \emph{reversible} process, as well as for finding an error in a prior version of this work. We thank Kai-Min Chung and Jamie Sikora for related discussions, and David Mestel for observing that the bound of Equation~(\ref{eqn:approx}) is not asymptotically optimal. AB acknowledges support by the U.S.~Air Force Office of Scientific Research under award number FA9550-17-1-0083, Canada’s NSERC, an Ontario ERA, and the University of Ottawa's Research Chairs program. SG acknowledges support by NSF grants CCF-1526189 and CCF-1617710. HSZ acknowledges support by NSF grant CNS-1801470 and a Google Faculty Research Award.

\appendix

\section{Universal Composition (UC) Framework}
\label{appendix:UCmodels}

We consider simulation-based security.
The Universal Composability (UC) framework was proposed by Canetti~\cite{FOCS:Canetti01}, culminating a long sequence of
simulation-based security definitions (\emph{c.f.} \cite{STOC:GolMicWig87,C:GolLev90,C:MicRog91,JC:Beaver91,JC:Canetti00}); please see also \cite{PW01,STOC:PraSah04,TCC:CDPW07,STOC:LinPasVen09,ITCS:MauRen11} for alternative/extended frameworks. Recently Unruh~\cite{EC:Unruh10} extend the UC framework to the quantum setting.
 Next, we provide a high-level description of the original classical UC model by Canetti~\cite{FOCS:Canetti01}, and then the quantum UC model by Unruh~\cite{EC:Unruh10}.

\subsection{Classical UC Model (\cite{FOCS:Canetti01})}

\paragraph{Machines.} The basic entities involved in the UC model are players $P_1, \ldots , P_k$ where $k$ is polynomial of security parameter $\secp$, an adversary $\cA$, and an environment $\cZ$. Each entity is modeled as a interactive Turing machine (ITM), where $\cZ$ could have an additional non-uniform string as advice. Each $P_i$ has identity $i$ assigned to it, while $\cA$ and $\cZ$ have special identities $id_{\cA}: = \tt{adv}$ and $id_{\cZ}:= {\tt env}$.

\paragraph{Protocol Execution.} A protocol specifies the programs for each $P_i$, which we denote as $\pi=(\pi_1,\ldots,\pi_k)$. The execution of a protocol is coordinated by the environment $\cZ$. It starts by preparing inputs to all players, who then run their respective programs on the inputs and exchange messages of the form $(id_{\tt sender}, id_{\tt receiver}, {\tt msg})$. $\cA$ can corrupt an arbitrary set of players and control them later on. In particular, $\cA$ can instruct a corrupted player sending messages to another player and also read messages that are sent to the corrupted players. During the course of execution, the environment $\cZ$ also interacts with~$\cA$ in an arbitrary way. In the end, $\cZ$ receives outputs from all the other players and generates one bit output. We use $\exec[\cZ,\cA, \pi]$ to denote the distribution of the environment $\cZ$'s (single-bit) output when executing protocol $\pi$ with $\cA$ and the $P_i$'s.

\paragraph{Ideal Functionality and Dummy Protocol.}
Ideal functionality $\func$ is a trusted party, modeled by an ITM again, that perfectly implements the desired multi-party computational task. We consider an ``dummy protocol'', denoted $\pidum^\cF$, where each party has direct communication with $\func$, who accomplishes the desired task according to the messages received from the players. The execution of $\pidum^\cF$ with environment $\cZ$ and an adversary, usually called the simulator $\cS$, is defined analogous as above, in particular, $\cS$ monitors the communication between corrupted parties and the ideal functionality $\func$. Similarly, we denote $\cZ$'s output distribution as $\exec[\cZ, \cS, \pidum^{\cF}]$.

\begin{definition}[Classical UC-secure Emulation] \
We say $\pi$ (classically) UC-emulates $\pi'$ if for any adversary~$\cA$, there exists a simulator $\cS$ such that for all environments $\cZ$,
\begin{equation}
\exec[\cZ,\cA, \pi]\approx \exec[\cZ,\cS, \pi']
\end{equation}
We here consider that $\cA$ and $\cZ$ are computationally unbounded, and we call it statistical UC-security. We require the running time $\cS$ is polynomial in that of $\cA$. We call this property {\em Polynomial Simulation}.
\end{definition}

 Let $\cF$ be a well-formed two party functionality. We say $\pi$ (classically) UC-realizes $\cF$ if for all adversary $\cA$, there exists a simulator $\cS$ such that for all environments $\cZ$, $\exec[\cZ,\cA, \pi]\approx \exec[\cZ,\cS,\pidum^\cF]$.
We also write $\exec[\cZ,\cA, \pi]\approx \exec[\cZ,\cS,\cF]$ if the context is clear.

UC-secure protocols admit a general composition property, demonstrated in the following universal composition theorem.
\begin{theorem}[{UC Composition Theorem~\cite{FOCS:Canetti01}}]
Let $\pi, \pi'$ and $\sigma$ be $n$-party protocols. Assume that $\pi$ UC-emulates $\pi'$. Then
$\sigma^{\pi}$ UC-emulates $\sigma^{\pi'}$. \label{thm:cuc}
\end{theorem}

\subsection{Quantum UC Model (\cite{EC:Unruh10})}
\label{section:qUC}

 Now, we give a high-level description of quantum UC model by Unruh~\cite{EC:Unruh10}.

\paragraph{Quantum Machine.} In the quantum UC model, all players are modeled as quantum machines. A quantum machine is a sequence of quantum circuits $\{M^\secp\}_{\secp\in \mathbb{N}}$, for each security parameter $\secp$. $M^\secp$ is a completely positive trace preserving operator on space $\cH^{\tt state}\otimes\cH^{\tt class}\otimes \cH^{\tt quant}$, where $\cH^{\tt state}$ represents the internal workspace of $M^\secp$ and $\cH^{\tt class}$ and $\cH^{\tt quant}$ represent the spaces for communication, where for convenience we divide the messages into classical and quantum parts. We allow a non-uniform quantum advice\footnote{Unruh's model only allows classical advice, but we tend to take the most general model. It is easy to justify that almost all results remain unchanged, including the composition theorem. See~\cite[Section 5]{C:HalSmiSon11} for more discussion.} to the machine of the environment $\cZ$, while all other machines are uniformly generated.

\paragraph{Protocol Execution.} In contrast to the communication policy in classical UC model, we consider a network~$\mathbf{N}$ which contains the space $\cH_{\mathbf{N}}:= \cH^{\tt class} \otimes \cH^{\tt quant}\otimes_i \cH^{\tt state}_i$. Namely, each machine maintains individual internal state space, but the communication space is shared among all . We assume $\cH^{\tt class}$ contains the message $(id_{\tt sender}, id_{\tt receiver}, {\tt msg})$ which specifies the sender and receiver of the current message, and the receiver then processes the quantum state on $\cH^{\tt quant}$. Note that this communication model implicitly ensures authentication. In a protocol execution, $\cZ$ is activated first, and at each round, one player applies the operation defined by its machine $M^\secp$ on $\cH^{\tt class} \otimes \cH^{\tt quant}\otimes \cH^{\tt state}$. In the end $\cZ$ generates a one-bit output. Denote $\exec[\cZ, \cA, \Pi]$ the output distribution of~$\cZ$.

\paragraph{Ideal Functionality.} All functionalities we consider in this work are classical, \emph{i.e.}, the inputs and outputs are classical, and its program can be implemented by an efficient classical Turing machine. Here in the quantum UC model, the ideal functionality $\func$ is still modeled as a quantum machine for consistency, but it only applies classical operations. Namely, it measures any input message in the computational basis to get a classical bit-string, and implements the operations specified by the classical computational task.

We consider an ``dummy protocol'', denoted $\pidum^\cF$, where each party has direct communication with $\func$, who accomplishes the desired task according to the messages received from the players. The execution of $\pidum^\cF$ with environment $\cZ$ and an adversary, usually called the simulator $\cS$, is defined analogous as above, in particular, $\cS$ monitors the communication between corrupted parties and the ideal functionality $\func$. Similarly, we denote $\cZ$'s output distribution as $\exec[\cZ, \cS, \pidum^{\cF}]$.  For simplicity, we also write it as $\exec[\cZ, \cS, {\cF}]$.

\begin{definition}[Quantum UC-secure Emulation]\
We say $\Pi$ quantum-UC-emulates $\Pi'$ if for any quantum adversary $\cA$, there exists a (quantum) simulator $\cS$ such that for all quantum environments $\cZ$,
\begin{equation}
\exec[\cZ,\cA, \Pi]\approx \exec[\cZ,\cS, \Pi']
\end{equation}
We consider here that
$\cA$ and $\cZ$ are computationally unbounded, we call it (quantum) statistical UC-security.
We require the running time $\cS$ is polynomial in that of $\cA$. We call this property {\em Polynomial Simulation}.
\label{def:quc}
\end{definition}

Similarly, (quantum) computational UC-security can be defined.
Let $\cF$ be a well-formed two party functionality. We say $\Pi$ {\bf quantum-UC-realizes} $\cF$ if for all quantum adversary $\cA$, there exists a (quantum) simulator $\cS$ such that for all quantum environments $\cZ$, $\exec[\cZ,\cA, \Pi]\approx \exec[\cZ,\cS,\pidum^\cF]$.

Quantum UC-secure protocols also admit general composition:
\begin{theorem}[{Quantum UC Composition Theorem~\cite[Theorem 11]{EC:Unruh10}}]
Let $\Pi, \Pi'$ and $\Sigma$ be quantum-polynomial-time
protocols. Assume that $\Pi$ quantum UC-emulates $\Pi'$. Then
$\Sigma^{\Pi}$ quantum UC-emulates $\Sigma^{\Pi'}$. \label{thm:quc}
\end{theorem}

\begin{remark}
Out of the two protocol parties (the sender and the receiver), we consider security only in
the case of the receiver being a corrupted party.
Note that we are only interested in cases where the same party
is corrupted with respect to all composed protocol. Furthermore, we only consider static corruption.

\end{remark}

\section{Stand-Alone Security in the case of a Malicious Sender}
\label{sec:appendix-def-malicious-sender}

In order to define stand-alone security against a malicious sender  (Definition~\ref{def:stand-alone-sender}), in our context, we closely follow definitions given in prior work~\cite{DNS10}, which we now recall. (Note that, instead of considering the \emph{approximate} case for security, we are able to use the \emph{exact}~one.)

\begin{definition}
 An $n$-step quantum two-party protocol with oracle calls, denoted $\Pi^\Oracle= (\PartyA, \PartyB, \Oracle, n)$ consists of:
 \begin{enumerate}
 \item input space $\mathcal{A}_0$ and $\mathcal{B}_0$ for parties $\PartyA$ and $\PartyB$ respectively.
 \item memory spaces $\mathcal{A}_1, \ldots \mathcal{A}_n$ and $\mathcal{B}_1, \ldots \mathcal{B}_n$ for $\PartyA$ and $\PartyB$, respectively.
 \item An $n$-tuple of quantum operations $(\PartyA_1, \ldots \PartyA_n)$ for $\PartyA$, $\PartyA_i : \spa{L}(\mathcal{A}_{i-1}) \mapsto \spa{L}(\mathcal{A}_{i}),   (1\leq i \leq n)$.

 \item An $n$-tuple of quantum operations $(\PartyB_1, \ldots \PartyB_n)$ for $\PartyB$, $\PartyB_i : \spa{L}(\mathcal{B}_{i-1}) \mapsto \spa{L}(\mathcal{B}_{i}),   (1\leq i \leq n)$.
 \item Memory spaces $\mathcal{A}_1, \ldots ,\mathcal{A}_n$ and $\mathcal{B}_1, \ldots ,\mathcal{B}_n$ can be written as $\mathcal{A}_i = {\mathcal{A}_i}^\Oracle \otimes  {\mathcal{A}_i}'$ and   $\mathcal{B}_i = {\mathcal{B}_i}^\Oracle \otimes  {\mathcal{B}_i}'$, $(1 \leq i \leq n)$ and $\Oracle = (\Oracle_1, \ldots , \Oracle_n)$ is an $n$-tuple of quantum operations: $\Oracle_i: \spa{L}(\mathcal{A}_i^\Oracle \otimes \mathcal{B}_i^\Oracle) \mapsto  \spa{L}(\mathcal{A}_i^\Oracle \otimes \mathcal{B}_i^\Oracle)$, $(1 \leq i \leq n)$.
 \end{enumerate}
\end{definition}

If $\Pi^\Oracle= (\PartyA, \PartyB, \Oracle, n)$ is  an $n$-turn two-party protocol, then the final state of the interaction upon input $\rho_{\text{in}} \in \mathrm{D}(\mathcal{A}_0 \otimes \mathcal{B}_0 \otimes \mathcal{R})$ where $\mathcal{R}$ is a system of dimension $\dim\mathcal{A}_0\dim\mathcal{B}_0$, is:
\begin{equation}
[\PartyA \oast \PartyB] (\rho_{\text{in}}) = (\mathbb{1}_{\spa{L}(\mathcal{A}'_{n} \otimes \mathcal{B}'_{n} \otimes \mathcal{R})}  \otimes \Oracle_n) (\PartyA_n \otimes \PartyB_n \otimes \mathbb{1}_\mathcal{R}) \ldots  (\mathbb{1}_{\spa{L}(\mathcal{A}'_{1} \otimes \mathcal{B}'_{1} \otimes \mathcal{R})}  \otimes \Oracle_1) (\PartyA_1 \otimes \PartyB_1 \otimes \mathbb{1}_\mathcal{R})(\rho_{\text{in}})\,.
\end{equation}

As in \cite{DNS10}, we specify that an oracle $\Oracle$ can be a communication oracle or an ideal functionality oracle.

An \emph{adversary} $\tilde{\PartyA}$ for an honest party $\PartyA$ in $\Pi^\Oracle= (\PartyA, \PartyB, \Oracle, n)$ is an $n$-tuple of quantum operations matching the input and outputs spaces of $\PartyA$. A  \emph{simulator} for $\tilde{\PartyA}$
is a sequence of quantum operations $(\Sim_i)_{i=1}^n$ where $\Sim_i$ has the same input-output spaces as the maps of $\tilde{\PartyA}$ at step~$i$.
In addition, $\Sim$ has access to the ideal functionality for the protocol $\Pi$.


\begin{definition}
\label{def:stand-alone-sender}
  An $n$-step quantum two-party protocol with oracle calls, $\Pi^\Oracle= (\PartyA, \PartyB, \Oracle, n)$ is statistically \emph{stand-alone} secure against a corrupt $\PartyA$ if for every adversary $\tilde{\PartyA}$ there exists a simulator $\Sim$ such that for every input $\rho_{\text{in}}$,
  \begin{equation}
     \trace_{\mathcal{B}_n \otimes \mathcal{R}} (\tilde{\PartyA} \oast \PartyB) =  \trace_{\mathcal{B}_n \otimes \mathcal{R}} (\Sim \oast \PartyB)\,.
  \end{equation}
\end{definition}

\noindent We note that Definition~\ref{def:stand-alone-sender} is weaker than some other definitions for active security used in the literature, \emph{e.g.},~\cite{DNS12}, because we ask only that the \emph{local} view of the adversary be simulated.

Given the simple structure of our protocol and ideal functionality, the construction and proof of the simulator is straightforward as shown below.

\begin{theorem}\label{thm:security-sender-stand-alone}
Protocol $\Pi$ is statistically stand-alone secure against a corrupt sender.
\end{theorem}

\begin{proof}
Since $\Pi$ consists in a single message from the sender to the receiver (together with a call to the ideal functionality for the token), we have  that $\PartyA = (\PartyA_1)$. Furthermore, since the ideal functionality $\funcHT$ does not return anything to the sender, there is no need for our simulator $\Sim$ to call an ideal functionality.

We thus build  $\Sim$ that runs $\PartyA$ on the input in register~$\mathcal{A}_0$. When $\PartyA$ calls the $\funcHT$ ideal functionality, the simulator does nothing. Since $\Pi$ is a one-way protocol, and since the ideal functionality also does not allow communication from the receiver to the sender,

\begin{equation}
\trace_{\mathcal{B}_n \otimes \mathcal{R}} (\tilde{\PartyA} \oast \PartyB) = \PartyA(\Tr_{\mathcal{B}_0 \otimes \mathcal{R} }(\rho_{\text{in}})) =
 \Sim(\Tr_{\mathcal{B}_0 \otimes \mathcal{R} } (\rho_{\text{in}}))\,.
\end{equation}
This concludes the proof.
\end{proof}

\section{Security Analysis for the Token}\label{scn:newproof}

We now provide the technical result (Theorem~\ref{thm:securetoken}) that is used to prove security of our Quantum OTM construction of Section~\ref{section:qOTM} against a linear number of queries. The statement below is informal; as outlined in Section~\ref{sscn:securityintuition}, to make it formal, in Section~\ref{sscn:linearsecurity} we model the user's interaction with the token via the Gutoski-Watrous (GW) framework for quantum games~\cite{GutoskiW07}. The resulting formal statement we desire, which immediately yields the informal claim below, is given in Theorem~\ref{thm:cheatingbound}.

 \begin{theorem}[Informal]
    For any stateless hardware token implemented as in Program~\ref{hardware-token-program}, \emph{i.e.}, using an $n$-qubit conjugate coding state $\ket{x}_\theta$, and for any user of the token (restricted only by the laws of quantum mechanics, meaning using any trace-preserving completely positive maps desired, regardless of efficiency of their implementation) making $m$ queries to the token, the probability the user successfully queries the token to extract both secret bits $s_0$ and $s_1$ is at most $O(2^{2m-0.228n})$.
\end{theorem}
\noindent Thus, we are able to prove that if the user makes at most $m=cn$ queries with $c<0.114$, then the user's probability of cheating successfully is exponentially small in $n$.

\subsection{Security against a linear number of token queries: Primal SDP}\label{sscn:linearsecurity}

To show security of our hardware token implementation (Program~\ref{hardware-token-program}) against a linear number of queries, we now model a user's interaction with the token as an interactive game between two parties using the GW framework of Section~\ref{sscn:gw}. As outlined in Section~\ref{sscn:securityintuition}, we shall treat the token as the \emph{co-strategy} and the user as the \emph{strategy}. An overview of how all operators introduced below fit together is given in Figure~\ref{fig:list}, which may be periodically referred to as the reader progresses through this section.

\paragraph{Basics of our model.} We proceed as follows. As depicted in Figure~\ref{fig:interaction}, the token (co-strategy) begins by preparing state $\rho_0\in\mathcal{L}(\spa{X}_1\otimes\spa{W}_0)$, and sending message $\spa{X}_1$ (which contains $\rho_R$ from Equation~(\ref{eqn:rhoR}) to the user. The user then makes $m$ queries, each via a distinct register $\spa{Y}_i$ for $i\in\set{1,\ldots, m}$. For each query made, we model the token as returning two strings: (1) a symbol in set $\Sigma=\set{0,1,\overline{0},\overline{1}}$ where $0$ and $1$ denote successful $0$- and $1$-queries, respectively, and $\overline{0}$ and $\overline{1}$ denote unsuccessful $0$- and $1$-queries, respectively, and (2) a bit $b$ which is set to $0$ for a failed query, or secret bit $b_i$ for a successful $i$th query. Formally, the size of each register $\spa{X}_i$ for $i\geq 2$ is hence three qubits. We will deviate from Figure~\ref{fig:interaction} in one respect --- we assume the token also returns the response to the final query, $m$, via a register $\spa{Y}_{m+1}$; this does not affect the success or failure of the user (as the latter makes no further queries at this point), but helps streamline the analysis. After this last response is sent out, the token measures the string $s\in\Sigma^m$ of responses it sent back to the user, and ``accepts'' if and only if $s$ contains at least one $0$ and one $1$. This will be spelled out formally below once we defined the isometries $A_i$ for the protocol.

Before doing so, let us introduce the terminology used in this section for discussing the secret key held by the token. Namely, recall in Program~\ref{hardware-token-program} that the token keeps secret key data $x\in\set{0,1}^n$ and $\theta\in\set{+,\times}^n$. Here, we shall replace these by a single string $z\in\set{0,1}^{2n}$, such that bits $2i$ and $2i+1$ of $z$ specify the basis and value of conjugate coding qubit $i$, for $i\in\set{1,\ldots, n}$ (i.e. $z_{2i}=\theta_{i}$ and $z_{2i+1}=x_i$). We shall call $z$ the \emph{secret key}. For consistency, we shall rename the \emph{quantum key} $\ket{x}_\theta$ from Program~\ref{hardware-token-program} by $\ket{\psi_z}\in(\complex^2)^{\otimes n}$, i.e. $\ket{x}_\theta=\ket{\psi_z}$. Next, in Program~\ref{hardware-token-program} the token takes inputs $b\in\set{0,1}$ and $y\in\set{0,1}^n$, for $b$ the choice bit and $y$ the claimed measured value. In this section, we shall simply concatenate these as one string $\widetilde{y}=b\circ y\in\set{0,1}^{n+1}$ (we henceforth write $\widetilde{y}=by$ for brevity), the first bit of which is the choice bit. We shall refer to $\widetilde{y}$ as a \emph{query string}. With these definitions in hand, for each secret key $z\in\set{0,1}^{2n}$, we define a partition $\rejz(z)$, $\rejo(z)$, $\accz(z)$, $\acco(z)$ of $\set{0,1}^{n+1}$, which correspond to the sets of query strings $\widetilde{y}$ which cause the token to return response $\overline{0}$, $\overline{1}$, $0$, or $1$, respectively.

\paragraph{Defining the isometries $A_k$.} Recall from Section~\ref{sscn:gw} that the GW model begins by capturing the actions of a co-strategy as a sequence of linear isometries, $A_k$. To define these $A_k$, we first construct operators $\Delta_k(z):\spa{Y}_k\mapsto\spa{X}_{k+1}\otimes\spa{W}_{k,k+1}$ (i.e. which map an incoming message in $\spa{Y}_k$ to the token to an outgoing message in $\spa{X}_{k+1}$ and private data to store in $\spa{W}_{k,k+1}$) for $k\in\set{1,\ldots, m}$ as follows:
\begin{align}
    \Delta_k(z) =& \sum_{\widetilde{y}\in\rejz(z)}\ket{\overline{0}0}_{\spa{X}_{k+1}}\ket{\widetilde{y}\overline{0}}_{\spa{W}_{k,k+1}}\bra{\widetilde{y}}_{\spa{Y}_k}+\label{eqn:independent}\\
    &\sum_{\widetilde{y}\in\rejo(z)}\ket{\overline{1}0}_{\spa{X}_{k+1}}\ket{\widetilde{y}\overline{1}}_{\spa{W}_{k,k+1}}\bra{\widetilde{y}}_{\spa{Y}_k}+\\
    &\sum_{\widetilde{y}\in\accz(z)}\ket{0s_0}_{\spa{X}_{k+1}}\ket{\widetilde{y}0}_{\spa{W}_{k,k+1}}\bra{\widetilde{y}}_{\spa{Y}_k}+\\
    &\sum_{\widetilde{y}\in\acco(z)}\ket{1s_1}_{\spa{X}_{k+1}}\ket{\widetilde{y}1}_{\spa{W}_{k,k+1}}\bra{\widetilde{y}}_{\spa{Y}_k}.
\end{align}
The intuition is as follows. In round $k$, we model the token as (coherently) making the following classical computation: Upon input $\ket{\widetilde{y}}_{\spa{Y}_k}$ from the user (which consists of a choice bit $b$ and candidate key $y$), the token sends its response in $\spa{X}_{k+1}$ to the user (the first symbol of which denotes accept/reject via a symbol from $\Sigma$, and the second symbol of which is the corresponding secret bit $s$, which is set to~$0$ by default for failed queries), and classically copies (i.e. via transversal CNOT gates) both the input $\widetilde{y}$ and response from $\Sigma$ into $\spa{W}_k$ (the private memory of the token). Recall from Section~\ref{sscn:securityintuition} that coherently keeping this local copy of $\widetilde{y}$, which is never accessed again, simulates a measurement of $\spa{Y}_k$ in the standard basis (by the principle of deferred measurement~\cite{NC00}). The response from $\Sigma$ is also locally stored in $\spa{W}_k$, solely for the token to be able to decide at the end of the protocol whether the user successfully extracted both secret bits. (More details on this below after the isometries $A_k$ are defined.)

    \begin{figure}[t]
    \centering
      \includegraphics[width=16cm]{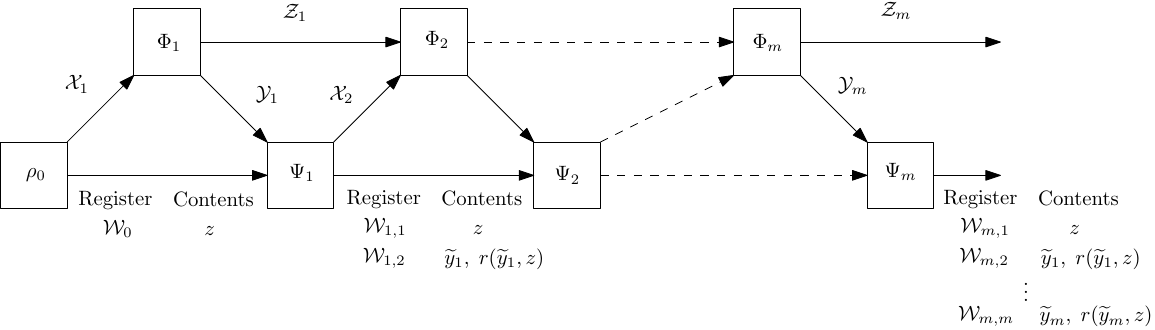}
      \caption{A reproduction of Figure~\ref{fig:interaction} with additional details regarding the token's private memory contents (bottom row, horizontal arrows pointing to the right) in each round. For example, in round $k=0$, $\spa{W}_0$ contains the secret key, $z$. In round $k=1$, $\spa{W}_{1,1}$ contains $z$ and $\spa{W}_{1,2}$ contains $\widetilde{y}_1$ and $r(\widetilde{y}_1,z)$. Here, $r(\widetilde{y}_k,z)\in\Sigma$ denotes whether the token accepted or rejected query string $\widetilde{y}_k$ assuming secret key $z$. Note the secret key $z$ is passed along from round to round (otherwise, the token cannot correctly decide its response in a round $k$ given query string $\widetilde{y}_k$).}
      \label{fig:interactionMemory}
    \end{figure}

Before finally defining isometries $A_k$, let us further elaborate on how the token's private memory spaces $\spa{W}_k$ is modelled (illustrated in Figure~\ref{fig:interactionMemory}).
\begin{itemize}
    \item $\spa{W}_0$ contains the secret key $z\in\set{0,1}^{2n}$ of the token (i.e. the token knows what the secret key is).
    \item Each $\spa{W}_k$ register for $k>0$ is split into $k+1$ parts:
    \begin{itemize}
        \item $\spa{W}_{k,1}$ contains a copy of $z$ (this allows us to pass forward $z$ from one round of interaction to the next, i.e. the token should know the secret key in \emph{all} rounds), and
        \item  $\spa{W}_{k,r}$ for $r\geq 2$ contains a copy in the standard basis of the user's $(r-1)$st query string (string $\widetilde{y}$), as well as the token's response from $\Sigma$ for said query.
    \end{itemize}
\end{itemize}
    \noindent Note that an additional technical reason for storing $\widetilde{y}$ above is that it ensures $\Delta_k(z)^\dagger\Delta_k(z)=I$, so that each $A_i$ defined shortly is an isometry. We remark that while the size of $\spa{W}$ grows with $m$ in our security analysis here, the actual token does not have growing memory requirements, since it stores nothing other than the secret key $z$ in its private memory (i.e. registers $\spa{W}_{k,r}$ for $r\geq 2$ exist only for our security analysis, not the actual implementation of the token).

\begin{figure}[t]
\centering
\begin{tabular}{c|l}
  \hline
  Operator & Description \\
  \hline
  &\\
   $\Delta_k(z)$    & 1. Sends back token's response to user's $k$th query $\widetilde{y}_k$, \emph{conditioned on} key $z$\\
                    & 2. Copies all data sent above to token's private register\\
                    & 3. Forwards token's existing private memory contents to next round\\
                    &\\
   \hline
   &\\
   $A_k$            & ``Bootstraps'' $\Delta_k(z)$ by summing over all possible keys $z$\\
                    & Note $A_0$ also has special role of sending quantum key $\ket{\psi_z}$ to user\\
   &\\
   \hline
   &\\
   $A$& The operator obtained by taking the product of all $A_k$\\
   &\\
   \hline
   &\\
   $B_1$ & The operator $A$ projected down onto the space of all ``accepting'' strategies, i.e. \\
   &where the token's private memory in the last round, $\mathcal{W}_m$, contains $\ket{0}$ and $\ket{1}$\\
    &in some $\mathcal{W}'_i$ and $\mathcal{W}'_j$ for $i\neq j$, respectively\\
   &\\
   \hline
    &\\
   $Q_1$& The operator $B_1$ is reshaped into a column vector via $\operatorname{vec}()$ mapping, then the token's \\
    &private memory in  the last round, $\mathcal{W}_m$, is traced out. \\
    &\\
  \hline
\end{tabular}
\caption{An overview of how all operators in our instantiation of the GW framework fit together.}
\label{fig:list}
\end{figure}

We are now ready to define isometries $A_k$ for round $k$ of the token's actions, where $1< k \leq m$:
\begin{align}
    A_0 &= \frac{1}{2^n}\sum_{z\in\set{0,1}^{2n}}\ket{\psi_z}_{\spa{X}_1}\ket{z}_{\spa{W}_{0,1}}\\
    A_1 &= \sum_{z\in\set{0,1}^{2n}}\Delta_1(z)_{\spa{Y}_1,\spa{X}_{2},\spa{W}_{1,2}}\otimes \ket{z}_{\spa{W}_{1,1}}\bra{z}_{\spa{W}_{0,1}}\\
    A_k &= \sum_{z\in\set{0,1}^{2n}}\Delta_k(z)_{\spa{Y}_k,\spa{X}_{k+1},\spa{W}_{k,k+1}}\otimes \ket{z}_{\spa{W}_{k,1}}\bra{z}_{\spa{W}_{k-1,1}} \bigotimes_{r=2}^{k}I_{\spa{W}_{k,r},\spa{W}_{k-1,r}}
\end{align}
Here, $A_0:\complex\mapsto \spa{X}_1\otimes \spa{W}_0$, and $A_k:\spa{Y}_k\otimes\spa{W}_{k-1}\mapsto\spa{X}_{k+1}\otimes\spa{W}_k$ for $1\leq k \leq m$. The intuition is as follows:
\begin{itemize}
    \item $A_0$ captures the token choosing an initial secret key $z$ uniformly at random and preparing corresponding quantum key $\ket{\psi_z}$, which it sends to the user in register $\spa{X}_1$.
    \item Each $A_i$ for $1\leq k\leq m$ consists of terms $\Delta_k(z)$ and $\ketbra{z}{z}$. The latter simply copies forward the secret key $z$ from round $i-1$ to $i$ from private register $\spa{W}_{k-1,1}$ to $\spa{W}_{k,1}$ , ensuring the token always knows $z$. Recall the term $\Delta_k(z)$, defined in Equation~(\ref{eqn:independent}), captures the token reading a message $\widetilde{y}$ from the user in $\spa{Y}_k$, measuring it in the standard basis (simulated by copying string $\widetilde{y}$ to a private register $\spa{W}_{k,k+1}$), returning an appropriate response to the user in register $\spa{X}_{k+1}$, and storing a copy of the $k$th response from $\Sigma$ to the user in the private register $\spa{W}_{k,k+1}$.
\end{itemize}
\begin{remark}\label{rem}
It is in the definition of the $A_i$ above that it is now formally clear that, although in our analysis the token stores additional data in its private register $\spa{W}$ (in addition to the private key, $z$), the token's response on incoming $k$th message $\widetilde{y}$ depends \emph{solely on register} $\spa{W}_{k,1}$, which contains only the secret key, $z$. (This is most easily seen through Equation~\ref{eqn:independent}, where the only ``bra'' vector is $\bra{\widetilde{y}}_{\spa{Y}_k}$, meaning the corresponding response $\ket{\overline{0}0}_{\spa{X}_{k+1}}$ depends only on $\widetilde{y}$ and the string $z$ (since the term $\overline{0}$ depends on the summation criterion $\widetilde{y}\in A_{\overline{0}}(z)$, which depends on $z$). Thus, the effective interactive behavior of the token is indeed stateless, as desired.
\end{remark}

\paragraph{Combining isometries $A_i$ to get $A$.} Having defined isometries $A_i$, their product now yields operator $A$ from Equation~(\ref{eqn:A}) (where we reorder the $\spa{X}$ and $\spa{W}$ registers to clarify that incoming message $\spa{Y}_k$ results in outgoing message $\spa{X}_{k+1}$):
\begin{align}
    A =& \frac{1}{2^n}\sum_{z\in\set{0,1}^{2n}}\sum_{\widetilde{y}_1,\ldots, \widetilde{y}_m\in \set{0,1}^{n+1}}\ket{\widetilde{y}_1r(\widetilde{y}_1,z)}_{\spa{W}_{m,2}}\otimes\cdots\otimes\ket{\widetilde{y}_mr(\widetilde{y}_m,z)}_{\spa{W}_{m,m+1}}\otimes \\
    &\ket{r(\widetilde{y}_{m},z)s_{r(\widetilde{y}_{m},z)}}_{\spa{X}_{m+1}}\bra{\widetilde{y}_m}_{\spa{Y}_m}\otimes \ket{r(\widetilde{y}_{m-1},z)s_{r(\widetilde{y}_{m-1},z)}}_{\spa{X}_{m}}\bra{\widetilde{y}_{m-1}}_{\spa{Y}_{m-1}} \otimes \cdots\otimes
     \\
    &\ket{r(\widetilde{y}_{1},z)s_{r(\widetilde{y}_{1},z)}}_{\spa{X}_{2}}\bra{\widetilde{y}_{1}}_{\spa{Y}_{1}}\otimes\ket{\psi_z}_{\spa{X}_1}\otimes\ket{z}_{\spa{W}_{m,1}},
\end{align}
where $r(\widetilde{y},z)\in\Sigma$ denotes whether the token accepted or rejected query string $\widetilde{y}$ assuming secret key $z$, and $s_{r(\widetilde{y},z)}\in\set{0,1}$ is the secret bit returned by the token corresponding to response $r(\widetilde{y},z)\in\Sigma$ (recall we set $s_{r(\widetilde{y},z)}=0$ if $r(\widetilde{y},z)\in\set{\overline{0},\overline{1}}$).

\paragraph{Defining operator $Q_1$.} In order to next define operator $Q_1$ from Equation~(\ref{eqn:CJ}), we model what it means for a cheating user of the token to ``succeed''. As mentioned earlier, this is formalized by having the token make a final measurement on its private memory after the protocol concludes, in order to determine whether the user has successfully extracted both secret bits via queries. Formally, for convenience, let $\spa{W}'$ denote the tensor product of the registers in $\spa{W}_{m,k}$ for $2\leq k\leq m+1$, which hold the values from $\Sigma$ (\emph{i.e.},~the responses $r(\widetilde{y}_{k-1},z)$). Then, a \emph{successful} user makes at least one correct $0$-query and at least one correct $1$-query (where a $j$-query refers to a query for choice bit $j$).

We define the ``accepting'' measurement operator $\Lambda_1$, corresponding to a successful user, as follows. $\Lambda_1$ maps $\spa{W}'$ to itself, and is a projector onto the set of strings with some $i\neq j$ such that $\spa{W}'_i$ is set to $\ket{0}$ and $\spa{W}'_j$ is set to $\ket{1}$. In other words, $\Lambda_1$ projects onto set
\begin{equation}\label{eqn:S}
    T:=\set{t\in \Sigma^m\mid t \text{ contains at least one $0$ and one $1$}}.
\end{equation}
To use this definition of $\Lambda_1$ to write down $B_1$, we require further terminology. Define for any $t\in T$ and fixed key $z\in\set{0,1}^{2n}$, the set of all consistent sequences of query strings $\widetilde{y}_i\in\set{0,1}^{n+1}$ as:
\begin{equation}
    Y_t = \set{(\widetilde{y},z)\in \set{0,1}^{m(n+1)}\times \set{0,1}^{2n}\mid r(\widetilde{y}_i,z)=t_i \text{ for $\widetilde{y}_i$ the $i$th block of $(n+1)$ bits in $\widetilde{y}$}}.
\end{equation}
(For example, the second block of $(n+1)$ bits of $0^{n+1}1^{n+1}$ is $1^{n+1}$.) In words, $Y_t$ is the set of all strings $\widetilde{y}_1\widetilde{y}_2\cdots \widetilde{y}_m z$ such that the response of the token on query $i$, $r(\widetilde{y}_i,z)\in\Sigma$, is consistent with $t_i$. Using this, define relation $R\subseteq\Sigma^m\times\set{0,1}^{m(n+1)}\times \set{0,1}^{2n}$ such that
\begin{equation}\label{eqn:R}
    (t,\widetilde{y},z)\in R \text{ if and only if }\left[ t\in T \text{ and } (\widetilde{y},z)\in Y_t\right].
\end{equation}
\noindent In words, a triple $(t,\widetilde{y},z)\in R$ if for a secret key $z$ and query string $\widetilde{y}$, $t\in T\subseteq \Sigma^m$ encodes the correct set of $m$ responses from the token (where recall $T$ is the set of all ``successful'' response strings).

Recall from Equation~(\ref{eqn:CJ}) that to define $Q_1$, we required $B_1$, which in turn required $\Lambda_1$ and $A$. With the latter two in hand, we can now define $B_1=(\sqrt{\Lambda_1}\otimes I)A=(\Lambda_1\otimes I)A$ as (where recall $t_i=r(\widetilde{y}_i,z)$)
\begin{align}
     B_1=& \frac{1}{2^n}\sum_{(t,\widetilde{y},z)\in R}\ket{\widetilde{y}_1t_1}_{\spa{W}_{m,2}}\otimes\cdots\otimes\ket{\widetilde{y}_m t_m}_{\spa{W}_{m,m+1}}\otimes \\
    &\hspace{18mm}\ket{t_{m}s_{t_{m}}}_{\spa{X}_{m+1}}\bra{\widetilde{y}_m}_{\spa{Y}_m}\otimes \ket{t_{m-1}s_{t_{m-1}}}_{\spa{X}_{m}}\bra{\widetilde{y}_{m-1}}_{\spa{Y}_{m-1}} \otimes \cdots\otimes
    \ket{t_1s_{t_1}}_{\spa{X}_{2}}\bra{\widetilde{y}_{1}}_{\spa{Y}_{1}}\otimes \\
    &\hspace{18mm}\ket{\psi_z}_{\spa{X}_1}\otimes\ket{z}_{\spa{W}_{m,1}}.
\end{align}
By definition of the $\operatorname{vec}$ mapping (Section~\ref{sscn:gw}),
\begin{align}
  \operatorname{vec}(B_1)=&\frac{1}{2^n}\sum_{(t,\widetilde{y},z)\in R}\ket{\widetilde{y}_1t_1}_{\spa{W}_{m,2}}\otimes\cdots\otimes\ket{\widetilde{y}_mt_m}_{\spa{W}_{m,m+1}}\otimes \\
    &\hspace{18mm}\ket{t_{m}s_{t_{m}}}_{\spa{X}_{m+1}}\ket{\widetilde{y}_m}_{\spa{Y}_m}\otimes \ket{t_{m-1}s_{t_{m-1}}}_{\spa{X}_{m}}\ket{\widetilde{y}_{m-1}}_{\spa{Y}_{m-1}} \otimes \cdots\otimes
    \ket{t_1s_{t_1}}_{\spa{X}_{2}}\ket{\widetilde{y}_{1}}_{\spa{Y}_{1}}\otimes \\
    &\hspace{18mm}\ket{z}_{\spa{W}_{m,1}}\otimes\ket{\psi_z}_{\spa{X}_1}.
\end{align}
Finally, $Q_1=\trace_{\spa{W}_m}(\operatorname{vec}(B_1)\operatorname{vec}(B_1)^*)$ equals
\begin{align}
  Q_1=&\frac{1}{2^{2n}}\sum_{(t,\widetilde{y},z)\in R}
     \ketbra{t_{m}s_{t_{m}}}{t_{m}s_{t_{m}}}_{\spa{X}_{m+1}}\otimes\ketbra{\widetilde{y}_m}{\widetilde{y}_m}_{\spa{Y}_m}\otimes\cdots\otimes\\
    &
    \hspace{19mm}\ketbra{t_1s_{t_1}}{t_1s_{t_1}}_{\spa{X}_{2}}\otimes\ketbra{\widetilde{y}_{1}}{\widetilde{y}_{1}}_{\spa{Y}_{1}}\otimes
    \ketbra{\psi_z}{\psi_z}_{\spa{X}_1}.
\end{align}

\noindent Note that we have crucially used the fact that queries to the token are \emph{classical strings}. Namely, since the token implicitly measures its input in the standard basis (modelled by copying each string $\widetilde{y}_i$ to register $\spa{W}_i$), the partial trace over $\spa{W}_m$ annihilates all cross terms in $\operatorname{vec}(B_1)\operatorname{vec}(B_1)^*$. Thus, $Q_1$ is conveniently simplified to a \emph{mixture} over $(t,\widetilde{y},z)\in R$, which is further block-diagonal with respect to all registers other than $\spa{X}_1$.

For convenience, we permute subsystems to rewrite:
\begin{align}
  Q_1=&\frac{1}{4^{n}}\sum_{t\in T}\ketbra{t_{m}s_{t_{m}}}{t_{m}s_{t_{m}}}_{\spa{X}_{m+1}}\otimes\cdots\otimes\ketbra{t_1s_{t_1}}{t_1s_{t_1}}_{\spa{X}_{2}}\otimes\\
  &\hspace{4mm}\left(\sum_{(\widetilde{y},z)\in Y_t}\ketbra{\widetilde{y}_m}{\widetilde{y}_m}_{\spa{Y}_m}\otimes \cdots \otimes
\ketbra{\widetilde{y}_1}{\widetilde{y}_1}_{\spa{Y}_{1}}\otimes \ketbra{\psi_z}{\psi_z}_{\spa{X}_1}\right).\label{eqn:qa}
\end{align}

\paragraph{The SDP.} Having set up all required operators for the GW framework, Equation~(\ref{eqn:final1}) of Section~\ref{sscn:gw} now yields the optimal probability with which a cheating user can succeed; we reproduce Equation~(\ref{eqn:final1}) below for convenience. Note the subsystem ordering of $Q_1$ below is not that of Equation~(\ref{eqn:qa}), but rather $Q_1\in \operatorname{Pos}(\spa{Y}_{1,\ldots, m}\otimes \spa{X}_{1,\ldots, m+1})$ below; we have omitted explicitly including the permutation effecting this reordering to avoid clutter. Also, to account for the slight asymmetry in our protocol (the token sends out $m+1$ messages $\spa{X}_i$, whereas the user only sends $m$ messages $\spa{Y}_i$), we add a dummy space $\spa{Y}_{m+1}=\complex$ which models an empty $(m+1)$st message from the user to the token.
		\begin{align}
			\text{min:}\quad & p\\
  		\text{subject to:}\quad & Q_1 \preceq p R_{m+1}\\
  		& R_k= P_k\otimes I_{\spa{Y}_k}\qquad\qquad\qquad &\text{for }1\leq k\leq m+1\\
        & \trace_{\spa{X}_k}(P_k)= R_{k-1} &\text{for }1\leq k\leq m+1\\
        & R_0 = 1\\
        & R_k\in \operatorname{Pos}(\spa{Y}_{1,\ldots, k}\otimes \spa{X}_{1,\ldots, k})&\text{for }1\leq k\leq m+1\\
        & P_k\in \operatorname{Pos}(\spa{Y}_{1,\ldots, k-1}\otimes \spa{X}_{1,\ldots, k})&\text{for }1\leq k\leq m+1\\
        &p\in[0,1] &
          	\end{align}	

\noindent In the analysis below, we shall sometimes analyze the optimization above, which we shall denote $\Gamma'$. However, note that technically it is not yet an SDP due to the quadratic constraint $Q_1\preceq p R_{m+1}$. It is, however, easily seen to be equivalent to the following bona fide SDP $\Gamma$:
		\begin{align}
			\text{min:}\quad & p\label{eqn:11}\\
  		\text{subject to:}\quad & Q_1 \preceq R_{m+1}\label{eqn:12}\\
  		& R_k= P_k\otimes I_{\spa{Y}_k}\qquad\qquad\qquad &\text{for }1\leq k\leq m+1\label{eqn:13}\\
        & \trace_{\spa{X}_k}(P_k)= R_{k-1} &\text{for }1\leq k\leq m+1\label{eqn:14}\\
        & R_0 = p\label{eqn:15}\\
        & R_k\in \operatorname{Pos}(\spa{Y}_{1,\ldots, k}\otimes \spa{X}_{1,\ldots, k})&\text{for }1\leq k\leq m+1\label{eqn:16}\\
        & P_k\in \operatorname{Pos}(\spa{Y}_{1,\ldots, k-1}\otimes \spa{X}_{1,\ldots, k})&\text{for }1\leq k\leq m+1\label{eqn:17}
          	\end{align}	
Above and henceforth, we use terminology $\spa{T}_{1\cdots k}$ to denote the space $\spa{T}_1\otimes\cdots\otimes\spa{T}_k$.

\paragraph{Warmup: A ``trivial'' solution.} We mentioned in Section~\ref{sscn:securityintuition} that obtaining a solution to $\Gamma$ which obtains the trivial bound $p\leq 1$ is not trivial. (Sometimes with SDPs, a scaled identity operator gives a feasible solution obtaining the desired trivial bound on the objective value; this unfortunately does not work here.) Let us hence warm up by demonstrating a solution attaining the trivial bound $p\leq 1$.

\begin{lemma}\label{l:redundant}
    The SDP $\Gamma$ has a feasible solution with $p=1$.
\end{lemma}
\begin{proof}
Recall from Equation~(\ref{eqn:qa}) that
\begin{equation}
  Q_1=\frac{1}{4^{n}}\sum_{(t,\widetilde{y},z)\in R}\ketbra{t,s}{t,s}_{\spa{X}_{{m+1}\cdots 2}}\otimes
  \left(\ketbra{\widetilde{y}_m}{\widetilde{y}_m}_{\spa{Y}_m}\otimes \cdots \otimes
\ketbra{\widetilde{y}_1}{\widetilde{y}_1}_{\spa{Y}_{1}}\right)\otimes \ketbra{\psi_z}{\psi_z}_{\spa{X}_1},
\end{equation}
where $t\in T \subseteq \Sigma^m$ and $s\in \set{0,1}^m$ are the resulting query responses and secret bits, respectively. (Recall from Equation~(\ref{eqn:qa}) that, formally, we should write $s_t$, as each $s$ depends on $t$; to save clutter and space below, however, we drop the subscript.) Observe that any fixed $\widetilde{y}\in\set{0,1}^{m(n+1)}$ and $z\in\set{0,1}^{2n}$ determine a \emph{unique} query response string $t\in\Sigma^m$ (which may or may not be in $T$); denote this as $t(\widetilde{y},z)$. Therefore,
\begin{equation}
  Q_1=\frac{1}{4^{n}}\sum_{\substack{\widetilde{y},z\\\text{s.t. }t(\widetilde{y},z)\in T}}\ketbra{t(\widetilde{y},z),s}{t(\widetilde{y},z),s}_{\spa{X}_{{m+1}\cdots 2}}\otimes
  \left(\ketbra{\widetilde{y}_m}{\widetilde{y}_m}_{\spa{Y}_m}\otimes \cdots \otimes
\ketbra{\widetilde{y}_1}{\widetilde{y}_1}_{\spa{Y}_{1}}\right)\otimes \ketbra{\psi_z}{\psi_z}_{\spa{X}_1},
\end{equation}
for $T\subseteq\Sigma^m$ defined as in Equation~(\ref{eqn:S}). Let us drop the constraint that $t(\widetilde{y},z)\in T$, \emph{i.e.}~choose
\begin{equation}
  R_{m+1}=\frac{1}{4^{n}}\sum_{\widetilde{y},z}\ketbra{t(\widetilde{y},z),s}{t(\widetilde{y},z),s}_{\spa{X}_{{m+1}\cdots 2}}\otimes
  \left(\ketbra{\widetilde{y}_{m}}{\widetilde{y}_{m}}_{\spa{Y}_{m}}\otimes \cdots \otimes
\ketbra{\widetilde{y}_1}{\widetilde{y}_1}_{\spa{Y}_{1}}\right)\otimes \ketbra{\psi_z}{\psi_z}_{\spa{X}_1}.
\end{equation}
Clearly, $Q_1\preceq p\cdot R_{m+1}$ for $p=1$, since we added positive semidefinite terms to $Q_1$ to get $R_{m+1}$. Thus, if $R_{m+1}$ satisfies the remaining primal constraints, then it has objective function value $p=1$.

To see that $R_{m+1}$ satisfies the constraints, clearly $R_{m+1}$ has $I$ in register $\spa{Y}_{m+1}$ (recall $\spa{Y}_{m+1}=\complex$, so this just means $\spa{Y}_{m+1}$ is trivially set to $1$). Let us now trace out $\spa{X}_{m+1}$; we require that register $\spa{Y}_{m-1}$ now also contains the identity. For this, $\trace_{\spa{X}_{m+1}}(R_{m+1})$ equals:
\begin{equation}
  \frac{1}{4^{n}}\sum_{\widetilde{y}_m,\ldots,\widetilde{y}_1}\sum_{z}\ketbra{t_{m-1}(\widetilde{y},z)s_{m-1}}{t_{m-1}(\widetilde{y},z)s_{m-1}}_{\spa{X}_{{m}\cdots 2}}\otimes
  \left(\ketbra{\widetilde{y}_{m}}{\widetilde{y}_{m}}_{\spa{Y}_{m}}\otimes \cdots \otimes
\ketbra{\widetilde{y}_1}{\widetilde{y}_1}_{\spa{Y}_{1}}\right)\otimes \ketbra{\psi_z}{\psi_z}_{\spa{X}_1},
\end{equation}
where for brevity we use $t_{m-1}(\widetilde{y},z)s_{m-1}$ to denote the first $m-1$ queries. But since we discarded the $m$th symbol of $t(\widetilde{y},z)$, registers $\spa{Y}_m$ and $\spa{X}_1$ are now independent. Thus, bringing in the sum over $\widetilde{y}_m$,
\begin{align}
\trace_{\spa{X}_{m+1}}(R_{m+1})=&\frac{1}{4^{n}}\sum_{\widetilde{y}_{m-1},\ldots,\widetilde{y}_1}\sum_{z}\ketbra{t_{m-1}(\widetilde{y},z)s_{m-1}}{t_{m-1}(\widetilde{y},z)s_{m-1}}_{\spa{X}_{{m}\cdots 2}}\otimes\\
  &\left(I_{\spa{Y}_{m}}\otimes \ketbra{\widetilde{y}_{m-1}}{\widetilde{y}_{m-1}}_{\spa{Y}_{m-1}}\otimes \cdots \otimes
\ketbra{\widetilde{y}_1}{\widetilde{y}_1}_{\spa{Y}_{1}}\right)\otimes \ketbra{\psi_z}{\psi_z}_{\spa{X}_1}.
\end{align}
In a similar fashion, tracing out registers $\spa{X}_{m\cdots 2}$ will yield operator
\begin{equation}
\frac{1}{4^{n}}I_{\spa{Y}_{m+1\cdots 1}}\otimes \sum_{z}\ketbra{\psi_z}{\psi_z}_{\spa{X}_1}.
\end{equation}
Finally, tracing out $\spa{X}_1$ yields $I_{\spa{Y}_{m\cdots 1}}$, since there are $4^n$ possible quantum key states $\ket{\psi_z}$. Hence, $R_{m+1}$ is a feasible solution.
\end{proof}

\paragraph{An upper bound on the cheating probability.} We now give a feasible solution to SDP $\Gamma$ which yields the claimed security against a linear number of queries. Its proof of correctness relies on Lemma~\ref{l:eval}, which we state and prove first.

\begin{lemma}\label{l:eval}
    For $Q_1$ in Equation~(\ref{eqn:qa}), $\lmax(Q_1)=\frac{2}{4^n}\left(1+\frac{1}{\sqrt{2}}\right)^n$.
\end{lemma}
\begin{proof}
The factor of $4^{-n}$ in the claimed value for $\lmax(Q_1)$ comes from the $4^{-n}$ appearing in Equation~(\ref{eqn:qa}); we henceforth thus ignore this $4^{-n}$ term in this proof by redefining $Q_1$ as $4^nQ_1$. We shall also ignore the $b_i$ terms in $Q_1$, as they shall play no role in the analysis. Now, since $Q_1$ is block-diagonal (with respect to the standard basis) on registers $\spa{X}_2$,\ldots,$\spa{X}_{m+1}$, $\spa{Y}_1$,\ldots,$\spa{Y}_m$, it suffices to characterize the largest eigenvalue of any block. We shall say that any fixed $t\in T$ and $\widetilde{y}\in\set{0,1}^{m(n+1)}$ defines the $(t,\widetilde{y})$-block of $Q_1$. (Formally, the $(t,\widetilde{y})$-block of $Q_1$ is given by $\Pi_{t,\widetilde{y}}Q_1\Pi_{t,\widetilde{y}}$, where $\Pi_{t,\widetilde{y}}=\ketbra{t}{t}_{\spa{X}_{m+1\cdots 2}}\otimes\ketbra{\widetilde{y}}{\widetilde{y}}_\spa{Y}$.)

\paragraph{Lower bound.} We first show lower bound $\lmax(Q_1)\geq \frac{2}{4^n}(1+\frac{1}{\sqrt{2}})^n$. To do so, we demonstrate an explicit $t,\widetilde{y}$ such that the $(t,\widetilde{y})$-block has eigenvalue $\frac{2}{4^n}(1+\frac{1}{\sqrt{2}})^n$. Set $t=0^{m-1}1$ (note $t\in T$) and $\widetilde{y}=\widetilde{y}_1\ldots \widetilde{y}_m$ for $\widetilde{y}_1=\widetilde{y}_2=\cdots=\widetilde{y}_{m-1}$ and $\widetilde{y}_{m-1}\neq \widetilde{y}_m$ (note $\widetilde{y}_i\in\set{0,1}^{n+1}$), where the first bit of each of $\widetilde{y}_1,\ldots, \widetilde{y}_{m-1}$ is $0$, and the first bit of $\widetilde{y}_m$ is $1$. In words, we are modelling $m-1$ successful (and identical) $0$-queries in the $Z$-basis, followed by a single successful $1$-query in the $X$-basis. The question now is: Given $t$ and $\widetilde{y}$, how many $\ket{\psi_z}\in(\complex^2)^{\otimes n}$ exist such that $(t,\widetilde{y},z)\in R$?

To answer this, observe that the token enforces the following set of rules. Fix any $i\in\set{1,\ldots, m}$, and let $\ket{\widetilde{y}_i(j)}$ and $\ket{\psi_z(j)}$ denote the $j$th qubits of $\widetilde{y}_i$ and $\psi_z$, respectively. Then we have rules (where $H$ denotes the $2\times 2$ Hadamard matrix, and $\overline{b}$ denotes the complement of bit $b$):
\begin{enumerate}
    \item If $t_i=0$, then $\forall j\in\set{1,\ldots,n}$, either $\ket{\psi_z(j)}=\ket{\widetilde{y}_i(j)}$ or $\ket{\psi_z(j)}\in\set{\ket{+},\ket{-}}$.
    \item If $t_i=1$, then $\forall j\in\set{1,\ldots,n}$, either $\ket{\psi_z(j)}=H\ket{\widetilde{y}_i(j)}$ or $\ket{\psi_z(j)}\in\set{\ket{0},\ket{1}}$.
    \item If $t_i=\overline{0}$, then $\exists j\in\set{1,\ldots,n}$ such that $\ket{\psi_z(j)}=\ket{\overline{\widetilde{y}_i(j)}}$.
    \item If $t_i=\overline{1}$, then $\exists j\in\set{1,\ldots,n}$ such that $\ket{\psi_z(j)}=H\ket{\overline{\widetilde{y}_i(j)}}$.
\end{enumerate}
Recall now that we set $t_1=0$ and $t_m=1$, \emph{i.e.}~the first query was a successful $Z$-basis query and the last query was a successful $X$-basis query. Applying rules $1$ and $2$ above thus yields that for all indices $k$, $\ket{\psi_z(k)}\in\set{\ket{\widetilde{y}_1(k)}, H\ket{\widetilde{y}_m(k)}}$. Moreover, since $\widetilde{y}_1=\widetilde{y}_2=\cdots = \widetilde{y}_{m-1}$, it follows that for all $k$, both assignments for $\ket{\psi_z(k)}$ are consistent with $t$. We conclude that the $(t,\widetilde{y})$-block of $Q_1$ has the following operator in register $\spa{X}_1$:
\begin{equation}\label{eqn:sigma}
    \sigma=\bigotimes_{k=1}^n\left(\ketbra{\widetilde{y}_1(k)}{\widetilde{y}_1(k)}+H\ketbra{\widetilde{y}_m(k)}{\widetilde{y}_m(k)}H\right).
\end{equation}
But for any $b,c\in\set{0,1}$, $\lmax(\ketbra{b}{b}+H\ketbra{c}{c}H)=1+\frac{1}{\sqrt{2}}$ (see, e.g.,~\cite{MVW13}). Thus, $\lmax(\sigma)=(1+\frac{1}{\sqrt{2}})^n$, as claimed.

\paragraph{Upper bound.} We next show a matching upper bound of $\lmax(Q_1)\leq \frac{2}{4^n}(1+\frac{1}{\sqrt{2}})^n$ among all $(t,\widetilde{y})$-blocks. For any $t\in T$, there exist indices $i\neq j$ such that $\widetilde{y}_i$ and $\widetilde{y}_j$ are a successful $0$- and $1$-query, respectively. Without loss of generality, assume $i=1$ and $j=m$. Then, as in the previous case, rules $1$ and $2$ imply that:
\begin{equation}\label{eqn:allk}
    \forall k\in\set{1,\ldots,n}, \quad\ket{\psi_z(k)}\in\set{\ket{\widetilde{y}_1(k)}, H\ket{\widetilde{y}_m(k)}}.
 \end{equation}
 Consider now any $\widetilde{y}_i$ for $1<i<m$, and suppose without loss of generality that $\widetilde{y}_i$ is a $0$-query, \emph{i.e.}~its first bit is set to $0$. There are two cases to analyze:
\begin{itemize}
    \item (Case 1: $t_i=0$) In this case, both query $1$ and query $i$ are successful $0$-queries; thus, they must agree on \emph{all} secret key bits which were encoded in the $Z$ basis. It follows from Rule 1 that for any bit $k$ on which $\widetilde{y}_1$ and $\widetilde{y}_i$ disagree, the secret key must have encoded bit $k$ in the $X$-basis. In other words, $\ket{\psi_z(k)}=H\ket{\widetilde{y}_m(k)}$ in Equation~(\ref{eqn:allk}) (\emph{i.e.}~one of the two possibilities is eliminated). (If $\widetilde{y}_1=\widetilde{y}_i$, on the other hand, no such additional constraint exists.)

    \item (Case 2: $t_i=\overline{0}$) In this case, query $i$ is an unsuccessful $0$-query. By Rule $3$, there exists a bit $k$ on which $\widetilde{y}_1$ and $\widetilde{y}_k$ disagree, and whose corresponding secret key bit was encoded in the $Z$ basis. In other words, $\ket{\psi_z(k)}=\ket{\widetilde{y}_1(k)}$ in Equation~(\ref{eqn:allk}) (\emph{i.e.},~one of the two possibilities is eliminated).
\end{itemize}
The analysis for $\widetilde{y}_i$ being a $1$-query is analogous. We conclude that for any $(t,\widetilde{y})$-block of $Q_1$, the operator in register $\spa{X}_1$ is of the form of $\sigma$ from Equation~(\ref{eqn:sigma}), except that the some of the indices $k$ may contain an operator consisting of only $1$ summand (e.g. $\ketbra{\widetilde{y}_1(k)}{\widetilde{y}_1(k)}$ instead of $\ketbra{\widetilde{y}_1(k)}{\widetilde{y}_1(k)}+H\ketbra{\widetilde{y}_m(k)}{\widetilde{y}_m(k)}H$). Since the omitted summands are all positive semidefinite, however, we conclude the eigenvalue on any $(t,\widetilde{y})$-block is at most the eigenvalue of $\sigma$ from Equation~(\ref{eqn:sigma}), \emph{i.e.},~at most $\lmax(Q_1)\leq \frac{2}{4^n}(1+\frac{1}{\sqrt{2}})^n$, as claimed.
\end{proof}

We can now prove the main result of this section.

\begin{theorem}\label{thm:cheatingbound}
    The SDP $\Gamma$ has a feasible solution with $p
\in O(2^{2m-0.228n})$.
\end{theorem}
\begin{proof}
    As $Q_1$ in Equation~(\ref{eqn:qa}) is block-diagonal in registers $\spa{X}_2,\ldots,\spa{X}_{m+1}$, consider solution (for $T$ from Equation~(\ref{eqn:S}))
\begin{equation}
    R_{m+1}=\frac{1}{\abs{T}}\sum_{t\in T}\ketbra{t_m s_{t_m}}{t_ms_{t_m}}_{\spa{X}_{m+1}}\otimes\cdots\otimes\ketbra{t_1s_{t_1}}{t_1s_{t_1}}_{\spa{X}_{2}}\otimes I_{\spa{Y}_{1,\ldots, m}}  \otimes  \frac{1}{2^n}I_{\spa{X}_1}.
\end{equation}
(Aside: Recall that $\spa{X}_1$ is an $n$-qubit register above, hence the $2^n$ renormalization factor.) Note that
\begin{align}
  \abs{\Sigma^m} &= 4^m \\
  \set{t\in \Sigma^m\mid t\text{ does not contain a }0} &= 3^m \\
  \set{t\in \Sigma^m\mid t\text{ does not contain a }1} &= 3^m \\
  \set{t\in \Sigma^m\mid t\text{ does not contain a }0\text{ or a }1} &= 2^m.
\end{align}
Thus, by the inclusion-exclusion principle, $\abs{T}=4^m-2\cdot3^m+2^m$.

In order for $R_{m+1}$ to be feasible, we must pick $p$ such that $Q_1\preceq p R_{m+1}$. Since $Q_1$ is block-diagonal on registers $\spa{X}_2\cdots \spa{X}_{m+1}$, it suffices to identify its block with the largest eigenvalue. In fact, each corresponding block for $R_{m+1}$ has eigenvalue $(\abs{T}2^n)^{-1}$. Thus, we must choose $p$ such that
\begin{equation}
    \lmax(Q_1)\leq \frac{p}{\abs{T}2^n},
\end{equation}
or equivalently, due to the $4^{-n}$ factor in $Q_1$,
\begin{equation}
    p\geq \frac{\abs{T}}{2^n}\lmax\left(4^nQ_1\right).
\end{equation}
By Lemma~\ref{l:eval}, $\lmax(Q_1)=\frac{2}{4^n}\left(1+\frac{1}{\sqrt{2}}\right)^n$. Thus, we can set
\begin{equation}
    p= \frac{\abs{T}}{2^{n-1}}\left(1+\frac{1}{\sqrt{2}}\right)^n\approx \abs{T}\cdot 2^{(-0.228)n+1},
\end{equation}
and since $\abs{T}\in\Theta(4^m)$, the cheating probability satisfies $p
\in O(2^{2m-0.228n})$.
\end{proof}

\section{Simplifying the Gutoski-Watrous SDP and its dual}\label{app:cleaner}

\subsection{Streamlining the primal and dual}\label{sscn:stream}

We now simplify the general SDP (Equation~(\ref{eqn:11}) from the Gutoski-Watrous (GW) framework (note this simplification is independent of our particular application of the framework for OTMs, i.e. independent of $Q_1$), and derive its dual SDP. For convenience, we begin by reproducing the following definitions, including the SDP $\Gamma$ of Equation~(\ref{eqn:11}).
		\begin{align}
			\text{min:}\quad & p\label{eqn:111}\\
  		\text{subject to:}\quad & Q_1 \preceq R_{m+1}\label{eqn:121}\\
  		& R_k= P_k\otimes I_{\spa{Y}_k}\qquad\qquad\qquad &\text{for }1\leq k\leq m+1\label{eqn:131}\\
        & \trace_{\spa{X}_k}(P_k)= R_{k-1} &\text{for }1\leq k\leq m+1\label{eqn:141}\\
        & R_0 = p\label{eqn:151}\\
        & R_k\in \operatorname{Pos}(\spa{Y}_{1,\ldots, k}\otimes \spa{X}_{1,\ldots, k})&\text{for }1\leq k\leq m+1\label{eqn:161}\\
        & P_k\in \operatorname{Pos}(\spa{Y}_{1,\ldots, k-1}\otimes \spa{X}_{1,\ldots, k})&\text{for }1\leq k\leq m+1\label{eqn:171}
          	\end{align}	
\begin{align}
    (t,\widetilde{y},z)&\in R \text{ if and only if }\left[ t\in T \text{ and } (\widetilde{y},z)\in Y_t\right]\\
  Q_1&=\frac{1}{4^{n}}\sum_{t\in T}\ketbra{t_{m}s_{t_{m}}}{t_{m}s_{t_{m}}}_{\spa{X}_{m+1}}\otimes\cdots\otimes\ketbra{t_1s_{t_1}}{t_1s_{t_1}}_{\spa{X}_{2}}\otimes\\
  &\hspace{4mm}\left(\sum_{(\widetilde{y},z)\in Y_t}\ketbra{\widetilde{y}_m}{\widetilde{y}_m}_{\spa{Y}_m}\otimes \cdots \otimes
\ketbra{\widetilde{y}_1}{\widetilde{y}_1}_{\spa{Y}_{1}}\otimes \ketbra{\psi_z}{\psi_z}_{\spa{X}_1}\right).
\end{align}

\paragraph{Guiding example: $m=3$.} We explicitly run through the construction for the first non-trivial case, $m=3$ queries. The construction then generalizes straightforwardly to all $m\geq 2$. To begin, using the fact that $R_4=P_4$ (since $\spa{Y}_4= \complex$, due to the fact that we assumed message $m+1$ from the user to the token is empty), $\Gamma$ can be written:
       		\begin{align}
			\text{min:}\quad & \trace(P_1)\\
  		\text{subject to:}\quad & Q_1 -  P_4\preceq 0\label{eqn:00}\\
  		& -P_3\otimes I_{\spa{Y}_3}+\trace_{\spa{X}_4}(P_4)\preceq 0\label{eqn:01}\\
        & -P_2\otimes I_{\spa{Y}_2}+\trace_{\spa{X}_3}(P_3)\preceq 0\label{eqn:02}\\
        & -P_1\otimes I_{\spa{Y}_1}+\trace_{\spa{X}_2}(P_2)\preceq 0\label{eqn:03}
          	\end{align}	
          Above, we relaxed the equalities to inequalities\footnote{This is without loss of generality, as we briefly justify. Clearly, any feasible solution for equality constraints is also feasible for inequality constraints. For the converse direction, suppose a feasible solution for the inequality constraints satisfies $P_i\otimes I_{\spa{Y}_{i+1}}-\trace_{\spa{X}_{i+1}}(P_{i+1})=\Lambda_i\succeq 0$ for non-zero $\Lambda_i$; pick the smallest such $i$ satisfying this condition. Then, redefining $P_{i+1}':= P_{i+1}+ \ketbra{\phi}{\phi}_{\spa{X}_{i+1}}\otimes \Lambda_i\succeq 0$ for arbitrary unit vector $\ket{\phi}$ satisfies $P_i\otimes I_{\spa{Y}_{i+1}}-\trace_{\spa{X}_{i+1}}(P'_{i+1})=0$, as desired. Note we can recurse this trick now from constraint $i$ to $i+1$, since if $P_{i+1}\otimes I_{\spa{Y}_{i+2}}-\trace_{\spa{X}_{i+2}}(P_{i+2})\succeq 0$, then $P'_{i+1}\otimes I_{\spa{Y}_{i+2}}-\trace_{\spa{X}_{i+2}}(P_{i+2})\succeq 0$ (similarly for constraint $Q_1\preceq P_{m+1}$). Thus, we obtain a new feasible solution for which all inequality constraints (except $Q_1\preceq P_{m+1}$) hold with equality. Moreover, this process does not alter the assignment for $P_1$ (i.e. we never define $P'_1$); thus the objective function value remains unchanged.}, which intuitively makes it easier to guess feasible solutions to $\Gamma$. We also omitted the positive semidefinite constraints on all $P_i$, since\footnote{In our particular setting, it is clear that $Q_1\succeq 0$. However, more generally in the GW framework, the operators $\set{Q_a}$ defining a measuring co-strategy all satisfy $Q_a\succeq 0$.} $P_4\succeq Q_1\succeq 0$ implies $P_i\succeq 0$ for all $i$. We now follow the standard Lagrange approach for deriving the dual SDP (see, e.g.~\cite{BV04}). Labelling equations~\eqref{eqn:00},\eqref{eqn:01},\eqref{eqn:02},\eqref{eqn:03} with dual variables $Y_1,\ldots ,Y_4$, respectively, the primal variables in the Lagrange dual function can be isolated as follows:
    \begin{center}
\begin{tabular}{c|c}
  Primal variable & Factor \\
  \hline
  $P_4$ & $-Y_1+Y_2\otimes I_{\X_4}$ \\
  $P_3$ & $-\trace_{\Y_3}(Y_2)+Y_3\otimes I_{\X_3}$ \\
  $P_2$ & $-\trace_{\Y_2}(Y_3)+Y_4\otimes I_{\X_2}$ \\
  $P_1$ & $I_{\X_1}-\trace_{\Y_1}(Y_4)$
\end{tabular}
\end{center}
For clarity and as an example, this says the term $P_4(-Y_1+Y_2\otimes I_{\X_4})$ appears in the dual function. This yields dual SDP:
          		\begin{align}
			\text{max:}\quad & \trace(Y_1Q_1)\\
  		\text{subject to:}\quad & -Y_1+Y_2\otimes I_{\X_4}=0\label{eqn:obv}\\
  		& -\trace_{\Y_3}(Y_2)+Y_3\otimes I_{\X_3}= 0\\
        & -\trace_{\Y_2}(Y_3)+Y_4\otimes I_{\X_2}= 0\\
        & I_{\X_1}-\trace_{\Y_1}(Y_4)= 0\\
        & Y_1,Y_2,Y_3,Y_4\succeq 0
          	\end{align}	
Now we make the following simplifications: (1) Replace $Y_1$ with $Y_2\otimes I_{\X_4}$ (follows from Equation~(\ref{eqn:obv})), (2) drop the constraints $Y_3,Y_4\succeq 0$ (since they are implied by $Y_2\succeq 0$), and (3) relax the equalities to inequalities (which follows similar to the argument for the primal, except here we also require that we are maximizing with respect to $Y_2$ below). Hence, we obtain:
          \begin{align}
			\text{max:}\quad & \trace(Y_2\trace_{\X_4}(Q_1))\\
  		\text{subject to:}& -\trace_{\Y_3}(Y_2)+Y_3\otimes I_{\X_3}\succeq  0\\
        & -\trace_{\Y_2}(Y_3)+Y_4\otimes I_{\X_2}\succeq  0\\
        & I_{\X_1}-\trace_{\Y_1}(Y_4)\succeq  0\\
        & Y_2\succeq 0
          	\end{align}	
          Note the $Y_2\succeq 0$ \emph{cannot} be removed. Intuitively, this is because the constraint $I_{\X_1}-\trace_{\Y_1}(Y_4)\succeq  0$ alone does not imply $Y_4\succeq 0$. Rather, it is the constraint $Y_2\succeq 0$ which forces $Y_4\succeq 0$ here. Indeed, a sanity check in CVX for Matlab reveals removing the $Y_2\succeq 0$ incorrectly yields an unbounded SDP.

          We now repeat the process by taking the dual of the dual to arrive at a simplified primal as follows. (Note that the inequalities above now go in the other direction, since we are starting from the dual SDP.) Labelling the constraints above $R_1,\ldots, R_4$, we have factor table:
    \begin{center}
\begin{tabular}{c|c}
  Dual variable & Factor \\
  \hline
  $Y_2$ & $\trace_{\X_4}(Q_1)-R_1\otimes I_{\Y_3}+R_4$ \\
  $Y_3$ & $\trace_{\X_3}(R_1)-R_2\otimes I_{\Y_2}$ \\
  $Y_4$ & $\trace_{\X_2}(R_2)-R_3\otimes I_{\Y_1}$
\end{tabular}
\end{center}
This yields primal SDP (after omitting the redundant constraints $R_1,R_2,R_3,R_4\succeq 0$):%
           		\begin{align}
			\text{min:}\quad & \trace(R_3)\\
  		\text{subject to:}&
  		\trace_{\X_4}(Q_1)-R_1\otimes I_{\Y_3}\preceq 0\\
        & \trace_{\X_3}(R_1)-R_2\otimes I_{\Y_2}\preceq 0\\
        & \trace_{\X_2}(R_2)-R_3\otimes I_{\Y_1}\preceq0
          	\end{align}	
          Taking the dual of the primal now yields the previous dual; so it seems we are done. Relabelling variables for the primal and dual, we obtain the final $m=3$ primal and dual SDPs, respectively:\\
          \begin{center}
          \begin{minipage}{2.5in}
           		\begin{align*}
			\text{min:}\quad & \trace(P_1)\\
  		\text{subject to:}\quad&
  		\trace_{\X_4}(Q_1)-P_3\otimes I_{\Y_3}\preceq 0\\
        & \trace_{\X_3}(P_3)-P_2\otimes I_{\Y_2}\preceq 0\\
        & \trace_{\X_2}(P_2)-P_1\otimes I_{\Y_1}\preceq0
          	\end{align*}	
          \end{minipage}
          \hspace{5mm}
          \begin{minipage}{2.5in}
          \begin{align*}
			\text{max:}\quad & \langle Y_1,\trace_{\X_4}(Q_1)\rangle\\
  		\text{subject to:}\quad& -\trace_{\Y_3}(Y_1)+Y_2\otimes I_{\X_3}\succeq  0\\
        & -\trace_{\Y_2}(Y_2)+Y_3\otimes I_{\X_2}\succeq  0\\
        & -\trace_{\Y_1}(Y_3)+I_{\X_1}\succeq  0\\
        & Y_1\succeq 0
          	\end{align*}	
          \end{minipage}
          \end{center}

\paragraph{General case.} The derivation above straightforwardly extends to the case of arbitrary $m\geq 2$, yielding primal and dual SDPs:\\
\begin{center}
{\underline{Primal SDP}}
\end{center}
\vspace{-3mm}
\begin{align}
		\text{min:}\quad & \trace(P_1)&\label{eqn:30}\\
  		\text{s.t.}\quad&\trace_{\X_{m+1}}(Q_1)-P_m\otimes I_{\Y_m}\preceq 0&\label{eqn:31}\\
        & \trace_{\X_{i+1}}(P_{i+1})-P_i\otimes I_{\Y_i}\preceq 0  &\forall i\in\set{1,\ldots, m-1}\label{eqn:32}\\
\end{align}	
\begin{center}
{\underline{Dual SDP}}
\end{center}
\vspace{-3mm}
\begin{align}
		\text{max:}\quad & \langle Y_1,\trace_{\X_{m+1}}(Q_1)\rangle\label{eqn:40}\\
  		\text{s.t.}\quad& -\trace_{\Y_{m-i+1}}(Y_i)+Y_{i+1}\otimes I_{\X_{m-i+1}}\succeq  0 &\forall i\in\set{1,\ldots ,m}\label{eqn:41}\\
        & Y_1\succeq 0\label{eqn:43}
\end{align}	
\noindent where note for uniformity in stating the dual constraints, we define $Y_{m+1}:=1$ in the dual SDP.

\subsection{An approximately optimal dual solution?}\label{sscn:approx}

We now give a simple feasible solution $Y:=\set{Y_i}$ to the dual SDP, whose objective function value appears to scale roughly as one might expect, \emph{if} security were to hold for our OTM construction. While we can explicitly prove $Y$ is \emph{not} dual optimal (thus, it only yields a lower bound on the best cheating probability), we conjecture it is roughly optimal up to multiplicative factors (stated precisely in Conjecture~\ref{conj:only}), which would in turn imply security against subexponentially many queries to the token, as desired.

\paragraph{A candidate dual solution.} Recall that each $\spa{Y}_i$ register encodes a message from the receiver to the token, consisting of $n+1$ qubits. Let $d:=2^{n+1}$ denote the dimension of this space. Define solution $Y:=\set{Y_1,\ldots,Y_m}$ via:
\begin{equation}\label{eqn:dual}
    Y_i = \frac{1}{d^{m-i+1}}I_{\spa{Y}_{1\cdots m},\spa{X}_{1\cdots m}}.
\end{equation}
Note that $Y_1\succeq 0$ trivially, and that Equation~(\ref{eqn:41}) holds with equality for all $i\in\set{1,\ldots, m}$. Thus, $Y$ is a dual feasible solution. Moreover, it obtains objective function value
\begin{equation}\label{eqn:prob}
    \beta:=\frac{\trace(Q_1)}{d^m}=\frac{\abs{R}}{4^nd^m},
\end{equation}
where recall we defined relation $R$ in Equation~(\ref{eqn:R}) via
    \begin{equation}
        (t,\widetilde{y},z)\in R \text{ if and only if }\left[ t\in T \text{ and } (\widetilde{y},z)\in Y_t\right].
    \end{equation}

\paragraph{The cardinality of $R$.} To analyze $\beta$, we require an expression for $\abs{R}$, given as follows.

\begin{lemma}\label{l:R}
    \begin{equation}
        \abs{R}=\left(2^{m(n+1)+n}\right)\sum_{\alpha=0}^n\binom{n}{\alpha}\left[1-\left(1-\frac{1}{2^{\alpha+1}}\right)^m-\left(1-\frac{1}{2^{n-\alpha+1}}\right)^m+\left(1-\frac{1}{2^{\alpha+1}}-\frac{1}{2^{n-\alpha+1}}\right)^m\right]
    \end{equation}
\end{lemma}
\begin{proof}
    Recall again from Equation~(\ref{eqn:R}) that $R$ is defined via
    \begin{equation}
        (t,\widetilde{y},z)\in R \text{ if and only if }\left[ t\in T \text{ and } (\widetilde{y},z)\in Y_t\right],
    \end{equation}
    where $T$ is the set of successful query responses. For any quantum key $\ket{\psi_z}$, let $\alpha$ denote the number of qubits in $\ket{\psi_z}$ which are encoded in the $Z$ basis. Note that fixing $\alpha$ partitions $R$ into $n+1$ sets; let $R_\alpha$ denote the set in this partition corresponding to $\alpha$ $Z$-bits. We analyze each $R_\alpha$ independently first.

    \paragraph{Computing $\abs{R_\alpha}$ for fixed $\alpha$.} Fix any $0\leq \alpha\leq n$, and any secret key $z\in\set{0,1}^{2n}$ with precisely $\alpha$ bits in the $Z$-basis; denote the resulting subset of $R_{\alpha}$ by $R_{\alpha,z}$. We now ask: Conditioned on secret key $z$ and token response $t\in\Sigma=\set{0,1,\overline{0},\overline{1}}$, how many query strings $\widetilde{y}$ are consistent with $t$?
    \begin{itemize}
        \item Case 1: $t=0$. Since we have to get precisely all $\alpha$ bits correctly (i.e. those in the $Z$ basis), and we can get up to $n-\alpha$ bits in the $X$ basis incorrect, we have $2^{n-\alpha}$ choices for $\widetilde{y}$. (Note that the choice bit for $\widetilde{y}$ is forced to be $0$ since $t=0$.)
        \item Case 2: $t=1$. This is analogous to $t=0$, except now we can get the $\alpha$ $Z$-bits incorrect and the $X$-bits must be correct. Thus, there are $2^\alpha$ strings $\widetilde{y}$.
        \item Case 3: $t=\overline{0}$. Since the $X$ bits can be anything, and there is precisely one correct setting to the $Z$ basis bits, we have $2^{n-\alpha}(2^{\alpha}-1)=2^n-2^{n-\alpha}$ choices for $\widetilde{y}$.
        \item Case 4: $t=\overline{1}$. Analogous to the $t=\overline{0}$ case, we have $2^\alpha(2^{n-\alpha}-1)=2^n-2^\alpha$ strings $\widetilde{y}$.
    \end{itemize}
    As a sanity check, note that summing the four values obtained above yields precisely $d=2^{n+1}$ strings $\widetilde{y}$, which is the dimension of each register $\spa{Y}_i$, as desired.

    To now obtain an expression for $\abs{R_{\alpha,z}}$, recall that any set of $m$ queries is successful if it contains at least one successful $0$-query and one successful $1$-query. Thus, using the inclusion-exclusion formula (intuition to follow):
    \begin{align}
        \abs{R_{\alpha,z}}=2^{m(n+1)}&-\sum_{a=0}^m\binom{m}{a}(2^{\alpha})^a\left[\sum_{b=0}^{m-a}\binom{m-a}{b}\left(2^n-2^{n-\alpha}\right)^b\left(2^n-2^\alpha\right)^{m-a-b}\right]\\
        &-\sum_{a=0}^m\binom{m}{a}(2^{n-\alpha})^a\left[\sum_{b=0}^{m-a}\binom{m-a}{b}\left(2^n-2^{n-\alpha}\right)^b\left(2^n-2^\alpha\right)^{m-a-b}\right]\\
        &+\sum_{b=0}^m\binom{m}{b}\left(2^n-2^{n-\alpha}\right)^b\left(2^n-2^\alpha\right)^{m-b}.
    \end{align}
    Above, the first term is the set of all query strings on $m$ messages. The negative terms count the number of query strings with no successful $0$-queries and no successful $1$-queries, respectively. For the former, for example, we first choose $a$ positions in which to put the successful $1$-queries, and then distribute $\overline{0}$- and $\overline{1}$-queries among the remaining $m-a$ positions. The final, positive, term, counts the number of query strings with neither a successful  $0$- nor a successful $1$-query.  Inverting the binomial expansion $(a+b)^k=\sum_{l=0}^k\binom{k}{l}a^lb^{k-l}$, we can next write:
    \begin{align}
        \abs{R_{\alpha,z}}=2^{m(n+1)}&-\sum_{a=0}^m\binom{m}{a}(2^\alpha)^a\left[(2^n-2^{n-\alpha})+(2^n-2^\alpha)\right]^{m-a}\\
        &-\sum_{a=0}^m\binom{m}{a}(2^{n-\alpha})^a\left[(2^n-2^{n-\alpha})+(2^n-2^\alpha)\right]^{m-a}\\
        &+\left[(2^n-2^{n-\alpha})+(2^n-2^\alpha)\right]^m.
    \end{align}
    Applying the binomial expansion again yields
    \begin{align}
        \abs{R_{\alpha,z}}=2^{m(n+1)}&-\left[2^\alpha+ (2^n-2^{n-\alpha})+(2^n-2^\alpha)\right]^m\\
        &-\left[2^{n-\alpha}+(2^n-2^{n-\alpha})+(2^n-2^\alpha)\right]^m\\
        &+\left[(2^n-2^{n-\alpha})+(2^n-2^\alpha)\right]^m.
    \end{align}
    Collecting like terms and factoring out $2^{m(n+1)}$ yields
    \begin{equation}
        \abs{R_{\alpha,z}}=\left(2^{m(n+1)}\right)\left[1-\left(1-\frac{1}{2^{\alpha+1}}\right)^m-\left(1-\frac{1}{2^{n-\alpha+1}}\right)^m+\left(1-\frac{1}{2^{\alpha+1}}-\frac{1}{2^{n-\alpha+1}}\right)^m\right].
    \end{equation}
    Recall this was for any fixed secret key $z$. But for any fixed $0\leq \alpha\leq n$, there are precisely $\binom{n}{\alpha}2^\alpha2^{n-\alpha}=\binom{n}{\alpha}2^n$ choices of $z$ with $\alpha$ qubits encoded in the $Z$-basis. Thus,
    \begin{equation}
        \abs{R_{\alpha}}=\binom{n}{\alpha}\left(2^{m(n+1)+n}\right)\left[1-\left(1-\frac{1}{2^{\alpha+1}}\right)^m-\left(1-\frac{1}{2^{n-\alpha+1}}\right)^m+\left(1-\frac{1}{2^{\alpha+1}}-\frac{1}{2^{n-\alpha+1}}\right)^m\right].
    \end{equation}
    \paragraph{The final expression.} The claim follows since the sets $R_{\alpha}$ partition $R$, and so $\abs{R}=\sum_{\alpha=0}^n\abs{R_\alpha}$.
\end{proof}

\paragraph{Using Lemma~\ref{l:R} to heuristically bound $\beta$.} Recall that our goal is to understand the dual value $\beta$ from Equation~(\ref{eqn:prob}) obtained by our dual solution, which in turn gives us a lower bound on the optimal cheating probability. So let us get a sense of how $\abs{R}$ might scale asymptotically by deriving a heuristic approximation. To begin, applying the H\"{o}lder inequality to $\abs{R}$ in Lemma~\ref{l:R} yields
\begin{equation}
    \abs{R}\leq\left(2^{m(n+1)+2n}\right)\cdot\max_{\alpha}\left[1-\left(1-\frac{1}{2^{\alpha+1}}\right)^m-\left(1-\frac{1}{2^{n-\alpha+1}}\right)^m+\left(1-\frac{1}{2^{\alpha+1}}-\frac{1}{2^{n-\alpha+1}}\right)^m\right]
\end{equation}
We shall assume\footnote{This assumption appears to hold in numerical calculations over various values of $\alpha$.} the maximum is attained for $\alpha=n/2$. For this choice of $\alpha$, recalling that for large $x$,
\begin{equation}
    \left(1-\frac{1}{x}\right)^m\approx{e^{-\frac{m}{x}}},
\end{equation}
in the large $n$ limit the term in the square brackets above is approximately
$
 \left[1-2e^{-\frac{m}{2^{n/2+1}}}+e^{-\frac{m}{2^{n/2}}}\right].
$
Hence, we may bound
 \begin{align}
    \abs{R}&\leq\left(2^{m(n+1)+2n}\right)\left[1-2e^{-\frac{m}{2^{n/2+1}}}+e^{-\frac{m}{2^{n/2}}}\right]\\
    &\leq\left(2^{m(n+1)}\right)4^n\left[1-e^{-\frac{m}{2^{n/2}}}\right]\\
    &\approx d^m4^n\frac{m}{2^{n/2}},
\end{align}
for $m\ll n$, and where in the last line we used $d=2^{n+1}$. Plugging this into Equation~(\ref{eqn:prob}), we get precisely the type of behavior we want:
 \begin{equation}\label{eqn:approx}
    \beta=\frac{\abs{R}}{4^nd^m}\lesssim \frac{m}{2^{n/2}}.
 \end{equation}
 Thus, for polynomial $m$, the objective function value obtained by our dual solution from Equation~(\ref{eqn:dual})is exponentially small in the number of key bits, $n$ (under the heuristic approximations made in this derivation).

\paragraph{The dual solution is not optimal.} Naturally, this raises the question of whether our dual solution $Y$ from Equation~(\ref{eqn:dual}) is optimal. If it were, then a matching primal solution can in principle be found (it is easy to see that Slater's constraint qualification holds for the primal and dual, and so strong duality holds), and thus the optimal cheating probability would be approximately that given in Equation~(\ref{eqn:approx}).

Unfortunately, $Y$ is provably not dual optimal. Specifically, for $m=2$ and $n=1$ ($2$ queries, $1$ key bit), the primal optimal value is $\approx 0.85$ and for $m=3, n=1$, it is\footnote{With $m=3$ and $n=1$, it is trivial to break the OTM construction. Namely, first measure the quantum key $\ket{\psi_z}$ in the standard basis and make an honest $0$-query to extract the first secret bit. Then, since $\ket{\psi_z}$ is only $1$ qubit, we can use brute force to make two $1$-queries to the token with the only two possible candidate keys in the $X$-basis, $0$, or $1$.} $1$ (both numerical values calculated via CVX in Matlab). However, evaluating $\abs{R}$ in Lemma~\ref{l:R} for these values of $m$ and $n$ yields $\beta=0.25$ and $\beta=0.46875$, respectively.

Moreover, the heuristic and loose upper bound on the objective function value of $Y$ from Equation~(\ref{eqn:approx}) is asymptotically not optimal\footnote{We thank David Mestel for bringing this to our attention.}, since the naive cheating strategy in which an adversary independently measures each qubit of $\ket{\psi_z}$ in basis $\set{\cos\frac{\pi}{8}\ket{0}+\sin\frac{\pi}{8}\ket{1}, -\sin\frac{\pi}{8}\ket{0}+\cos\frac{\pi}{8}\ket{1}}$ successfully obtains classical key $x\in\set{0,1}^n$ (see Program~\ref{hardware-token-program}) with probability $(\cos^2\frac{\pi}{8})^n\approx 2^{-0.228n}$.

Thus, the dual solution $Y$ of Equation~(\ref{eqn:dual}) is not optimal. However, as the heuristic derivation of $\beta$ from Equation~(\ref{eqn:approx}) rather naturally led to the desired type of bound on the cheating probability, we conjecture that $Y$ is \emph{approximately} optimal, in the following sense.
\begin{conj}\label{conj:only}
    The optimal values for the primal and dual SDPs of Section~\ref{sscn:stream} are, up to multiplicative scaling by some function $f(m,n)\in O(m^c2^{(1/2-\epsilon)n})$ for constants $c>0$ and $0<\epsilon<1/2$, equal to $\beta=\frac{\abs{R}}{4^nd^m}$.
\end{conj}
\noindent If Conjecture~\ref{conj:only} holds, then our protocol would be secure in the sense that the optimal cheating probability would scale as $\poly(m)/2^{\Theta(n)}$.

\section{Proof of Lemma~\ref{l:impossible}}\label{app:4.1}


\begin{proof}
Observe first that an honest receiver Alice wishing to extract $s_i$ acts as follows.
She applies a unitary $U_i\in\unitary({A}\otimes{B})$ to get state
\begin{equation}\label{eqn:step1}
    \ket{\phi_1}:=U_i\ket{\psi}_{{A}{B}}\ket{0}_{{C}}.
\end{equation}
She then measures ${B}$ in the computational basis and postselects on result $y\in\set{0,1}^n$, obtaining state
\begin{equation}\label{eqn:step2}
    \ket{\phi_2}:=\ket{\phi_y}_{{A}} \ket{y}_{{B}} \ket{0}_{{C}}.
\end{equation}
She now treats $y$ as a ``key'' for $s_i$, \emph{i.e.}, she applies $O_f$ to ${B}\otimes{C}$ to obtain her desired bit $s_i$, \emph{i.e.},
\begin{equation}\label{eqn:step3}
    \ket{\phi_3}:=\ket{\phi_y}_{{A}} \ket{y}_{{B}} \ket{s_i}_{{C}}.
\end{equation}

A malicious receiver Bob wishing to extract $s_0$ and $s_1$ now acts similarly to the rewinding strategy for superposition queries. Suppose without loss of generality that $s_0$ has at most $\Delta$ keys. Then, Bob first applies $U_0$ to prepare~$\ket{\phi_1}$ from Equation~(\ref{eqn:step1}), which we can express as
    \begin{equation}\label{eqn:mid}
        \ket{\phi_1}=\sum_{y\in\set{0,1}^n}\alpha_y \ket{\psi_y}_{{A}}\ket{y}_{{B}}\ket{0}_{{C}}.
    \end{equation}
    for $\sum_y\abs{\alpha_y}^2=1$. Since measuring ${B}$ next would allow us to retrieve $s_0$ in register $C$ with certainty, we have that all $y$ appearing in the expansion above satisfy $f(y)=s_0$. Moreover, since $s_0$ has at most $\Delta$ keys, there exists a key $y'$ such that $\abs{\alpha_{y'}}^2\geq 1/\Delta$. Bob now measures $B$ in the computational basis to obtain $\ket{\phi_2}$ from Equation~(\ref{eqn:step2}), obtaining $y'$ with probability at least $1/\Delta$. Feeding $y'$ into $O_f$ yields $s_0$. Having obtained $y'$, we have that
    $
        \abs{\braket{\phi_1}{\phi_2}}^2\geq 1/\Delta,
    $
    implying
    \begin{equation}
        \abs{\bra{\psi}U_0^\dagger\ket{\phi_{y'}}\ket{y'}}^2\geq 1/\Delta,
    \end{equation}
    \emph{i.e.}, Bob now applies $U_0^\dagger$ to recover a state with ``large'' overlap with initial state $\ket{\psi}$.

    To next recover $s_1$, define $\psigood:=U_1\ket{\psi}$ and $\psiapprox:=U_1U_0^\dagger\ket{\phi_{y'}}\ket{y'}$. Bob applies $U_1$ to obtain
    \begin{equation}
        \psiapprox = \beta_1\psigood+\beta_2\ket{\psi_{\rm good}^\perp},
    \end{equation}
    where $\sum_i\abs{\beta_i}^2=1$, $\braket{\psi_{\rm good}}{\psi_{\rm good}^\perp}=0$, and $\abs{\beta_1}^2\geq 1/\Delta$. Define $\pigood:=\sum_{y\in\set{0,1}^n\text{ s.t. }f(y)=s_1}\ketbra{y}{y}$. Then, the probability that measuring $B$ in the computational basis now yields a valid key for $s_1$ is
    \begin{equation}
        \bra{\psi_{\rm approx}}\pigood\psiapprox\geq \abs{\beta_1}^2\geq\frac{1}{\Delta},
    \end{equation}
    where we have used the fact that $\pigood\psigood=\psigood$ (since an honest receiver can extract $s_1$ with certainty). We conclude that Bob can extract both $s_0$ and $s_1$ with probability at least $1/\Delta^2$.
\end{proof}

\printbibliography

\end{document}